\documentclass[a4paper, 11pt, abstract=true]{scrartcl}

\usepackage{amsmath}
\usepackage{amsthm}
\usepackage{amsfonts}
\usepackage{amssymb}
\usepackage{bbm}
\usepackage{graphicx}
\usepackage{mathtools}

\usepackage{listings}
\newtheorem{thm}{Theorem}[section]
\newtheorem{lem}[thm]{Lemma}
\newtheorem{prop}[thm]{Proposition}
\newtheorem{cor}[thm]{Corollary}

\newtheorem{dfn}[thm]{Definition}

\newcommand{\tr}{\operatorname{tr}}
\newcommand{\id}{\mathbbm{1}}

\newcommand{\diag}{\operatorname{diag}}

\newcommand{\N}{\mathbb{N}}
\newcommand{\R}{\mathbb{R}}
\newcommand{\C}{\mathbb{C}}
\newcommand{\Ca}{\mathcal{C}}

\usepackage{abstract}

\usepackage[backend=bibtex, style=alphabetic]{biblatex}
\bibliography{literatur}
\usepackage{url}
\usepackage{color}
\usepackage[plainpages=false,breaklinks=true,pdftitle={An operational measure for squeezing},pdfauthor={Martin Idel, Daniel Lercher, Michael Wolf},pdfsubject={Operational squeezing measures}]{hyperref}
\definecolor{linkblue}{rgb}{0.1,0.2,.7}
\definecolor{citegreen}{rgb}{0.1,0.7,.2}
\hypersetup{
    colorlinks,
    citecolor=citegreen,
    filecolor=linkblue,
    linkcolor=linkblue,
    urlcolor=linkblue
}
\DeclareFieldFormat{doi}{%
  \mkbibacro{DOI}\addcolon\space
    \ifhyperref
      {\href{http://dx.doi.org/#1}{\nolinkurl{#1}}}
      {\nolinkurl{#1}}}

\allowdisplaybreaks[1]

\title{\vspace{-15mm}%
	An operational measure for squeezing}
\author{%
	\large
	\textsc{Martin Idel, Daniel Lercher, Michael M. Wolf} \\[2mm]
	\normalsize	Zentrum Mathematik, Technische Universit\"{a}t M\"{u}nchen \\
	\vspace{-5mm}
	}
\date{}

\begin{document}

\maketitle

\begin{abstract}
We propose and analyse a mathematical measure for the amount of squeezing contained in a continuous variable quantum state. We show that the proposed measure operationally quantifies the minimal amount of squeezing needed to prepare a given quantum state and that it can be regarded as a squeezing analogue of the “entanglement of formation”. We prove that the measure is convex and superadditive and we provide analytic bounds as well as a numerical convex optimisation algorithm for its computation. By example, we then show that the amount of squeezing needed for the preparation of certain multi-mode quantum states can be significantly lower than naive approaches suggest.
\end{abstract}

\tableofcontents
	
\section{Introduction}
The interplay between quantum optics and the field of quantum information processing, in particular via the subfield of continuous variable quantum information, has been developing for several decades and is interesting also due to its experimental success (see \cite{kok10} for a thorough introduction). 

Coherent bosonic states and the broader class of Gaussian bosonic states, quantum states whose Wigner function is characterised by its first and second moments, are of particular interest in the theory of continuous variable quantum information. Their interest is also due to the fact that modes of light in optical experiments behave like Gaussian coherent states.

For any bosonic state, its matrix of second moments, the so called covariance matrix, must fulfil Heisenberg's uncertainty principle in all modes. If the state possesses a mode, where despite this inequality $\Delta x\Delta p\geq\hbar/2$ either $\Delta x $ or $\Delta p$ is strictly smaller than $\sqrt{\hbar/2}$, it is called \emph{squeezed}. The production of squeezed states is experimentally possible, but it requires the use of nonlinear optical elements \cite{bra05a}, which are more difficult to produce and handle than the usual linear optics (i.e. beam splitters and phase shifters). Nevertheless, squeezed states play a crucial role in many experiments in quantum information processing and beyond. Therefore, it is natural both theoretically and practically to investigate the amount of squeezing which is necessary to create an arbitrary quantum state. 

As a qualitative answer, squeezing is known to be an irreducible resource with respect to linear quantum optics \cite{bra05a}. In the Gaussian case, it is also known to be closely related to entanglement of states \cite{wol03} and the non-additivity of quantum channel capacities \cite{ler13}. In addition, quantitative measures of squeezing have been provided on multiple occasions \cite{lee88,kra03b}, yet none of these measures are operational for more than a single mode in the sense that they do not measure the minimal amount of squeezing necessary to prepare a given state. 

The goal of this paper is therefore twofold: First, we define and study operational squeezing measures, especially measures quantifying the amount of squeezing needed to prepare a given state. Second, we reinvestigate in how far squeezing is a resource in a mathematically rigorous manner and study the resulting resource theory by defining preparation measures.

In order to give a brief overview of the results, we assume the reader is familiar with standard notation of the field, which is also gathered in section \ref{sec:preliminaries}. In particular, let $\gamma$ denote covariance matrices. A squeezed state is a state where at least one of the eigenvalues of $\gamma$ is smaller than one. The value of the smallest eigenvalue has been taken as a measure for squeezing before \cite{kra03b}, however it is an extremely coarse measure, as it only accounts for one mode. To remedy this fact, one might naturally propose the following measure:
\begin{align*}
	G_{\mathrm{squeeze}}(\gamma)=\prod_{\lambda_i<1} \lambda_i(\gamma)^{-1},
\end{align*}
where the $\lambda_i$ are the eigenvalues of $\gamma$. However, it is not clear how to interpret this measure operationally, thus we proceed via a different route.

To obtain operational squeezing measures, we first study operational squeezing in section \ref{sec:mathmeasure}: Suppose we want to implement an operation on our quantum state corresponding to some unitary $U$. Any such unitary can be implemented as the time-evolution of Hamiltonians. Recall that any quantum-optical Hamiltonian can be split into ``passive'' and ``active'' parts, where the passive parts are implementable by linear optics and the active parts require nonlinear media. We assume that the active transformations available are single-mode squeezers with Hamiltonian 
\begin{align*}
	H_{\mathrm{squeeze},j}=i\frac{\hbar}{2}(a_j^2-a_j^{\dagger\,2})
\end{align*}
where the $j$ denotes squeezing in the $j$-th mode and the $c$ is a complex coefficient, which can be seen as the interaction strength of the medium. We therefore consider any Hamiltonian of the form
\begin{align}
	H=H_{\mathrm{passive}}(t)+\sum_i c_i(t)H_{\mathrm{squeeze},j} \label{eqn:hamilton}
\end{align}
with any passive Hamiltonian $H_{\mathrm{passive}}$. Then, a natural measure of the squeezing costs to implement this Hamiltonian would be given by
\begin{align*}
	f_{\mathrm{squeeze}}(H)=\int \sum_i |c_i(t)|\, \mathrm{d}t
\end{align*}
Our squeezing measure for the operation $U$ is then defined as the mimimum of $f_{\mathrm{squeeze}}(H)$ for all Hamiltonians implementing the operation $U$ of the form (\ref{eqn:hamilton}). With this definition, we have an operational measure answering the question: Given an operation $U$, what is the most efficient way (in terms of squeezing) to implement it using passive operations and single-mode squeezers?

Instead of working with the generators, which are unbounded operators and therefore introduce a lot of analytic problems, we will work on the level of Wigner functions and therefore with the symplectic group. The unitary $U$ then corresponds to a symplectic matrix $S$ and we prove that the most efficient way to implement it is by using the Euler decomposition, also known as Bloch-Messiah decomposition. We show this result first in the case where the functions $c_i$ are step functions and later on in the more general case of measurable $c$ (Section \ref{sec:liealg}). In particular, the result implies that the minimum amount of squeezing to implement the symplectic matrix $S\in \R^{2n\times 2n}$ is given by
\begin{align}
	F(S):=\sum_{i=1}^n \log s^{\downarrow}_i(S)
\end{align}
where $s^{\downarrow}_i$ denotes the $i$-th singular value of $S$ ordered decreasingly.

With this in mind, we define a squeezing measure for preparation procedures where one starts out with a covariance matrix of an unsqueezed state and then performs symplectic (and possibly other) operations to obtain the state. More precisely, we define
\begin{align}
	G(\gamma):=\inf\left\{\sum_{i=1}^n \log s^{\downarrow}_i(S)\middle|\gamma \geq S^TS,~S\in Sp(2n)\right\}.
\end{align}
One of the main results of this paper, which will be proven in section \ref{sec:operational}, is that this measure is indeed operational in that it quantifies the minimal amount of single-mode squeezing necessary to prepare a state with covariance matrix $\gamma$, using linear optics with single-mode squeezers, ancillas, measurements, convex combinations and addition of classical noise.

We also define a second squeezing measure, which is a squeezing-analogue of the entanglement of formation, the ``squeezing of formation'', i.e. the amount of single-mode squeezed resource states needed to prepare a given state using only passive operations and adding of noise. This is done in section \ref{sec:resourcemeasure}, where we also prove that this measure is equal to $G$. 

In addition, we prove several structural facts about $G$ in section \ref{sec:measureg}. In particular, $G$ is convex, lower semicontinuous everywhere, continuous on the interior and subadditive. Moreover, we show
\begin{align*}
	\frac{1}{2}\log G_{\mathrm{squeeze}}(\gamma)\leq G(\gamma)
\end{align*}
where equality in this lower bound is usually not achievable, albeit numerical tests have shown that the bound is often very good. 

The measure would lose a lot of its appeal, if it could not be computed. Although we cannot give an efficient analytical formula for more than one mode, we provide a numerical algorithm to obtain $G$ for any state. To demonstrate that this works in principle, we calculate $G$ approximately for a state studied in \cite{mis08} (section \ref{sec:examples}). The calculations also demonstrate that the preparation procedure obtained from minimizing $G$ can greatly lower the squeezing costs when compared to naive preparation procedures.

Finally, we critically discuss the flexibility and applicability of our measures in Section \ref{sec:discopen}. We believe that while we managed to give reasonable measures and interesting tools to study the resource theory of squeezing from a theoretical perspective, $G$ might not reflect the experimental reality in all parts. In particular, it becomes extraordinarily difficult to achieve high squeezing in a single mode \cite{and15}, which is not reflected by taking the logarithm of the squeezing parameter. We show that this shortcoming can be easily corrected for a broad class of cost functions. In addition, the form of the active part of the Hamiltonian (\ref{eqn:hamilton}) might not reflect the form of the Hamiltonian in the lab. This cannot be corrected as easily but in any case, our measure will give a lower bound. 

\section{Preliminaries} \label{sec:preliminaries}
In this section, we collect basic notions from continuous variable quantum information and symplectic linear algebra that we need in the process. For a broader overview, we refer the reader to \cite{ade14,bra05b}.
\subsection{Phase Space in Quantum Physics}
Consider a bosonic system with $n$-modes, each of which is characterised by a pair of canonical variables $\{Q_k,P_k\}$. Setting $R=(Q_1,P_1,\ldots, Q_n,P_n)^T$ the canonical commutation relations (CCR) take on the form $[R_k,R_l]=i\sigma_{kl}$ with the standard symplectic form
\begin{align*}
	\sigma=\bigoplus_{i=1}^n \begin{pmatrix}{} 0 & 1 \\ -1 & 0 \end{pmatrix}
\end{align*}
Since it will sometimes be convenient, we also introduce another basis of the canonical variables: Let $\tilde{R}=(Q_1,Q_2,\ldots,Q_n,P_1,P_2,\ldots,P_n)^T$, then the symplectic canonical commutation relations take on the form $[\tilde{R}_k,\tilde{R}_l]=iJ_{kl}$ with the symplectic form
\begin{align*}
	J=\begin{pmatrix} 0 & \id_n \\ -\id_n & 0 \end{pmatrix}.
\end{align*}
Clearly, $J$ and $\sigma$ differ only by a permutation, since $R$ and $\tilde{R}$ differ only by a permutation. 

From functional analysis, it is well-known that the operators $Q_k$ and $P_k$ cannot be represented by bounded operators on a Hilbert space. In order to avoid complications associated to unbounded operators, it is usually easier to work with a representation of the CCR-relations on some Hilbert space $\mathcal{H}$, instead. The standard representation is known as the \emph{Schr\"{o}dinger representation} and defines the \emph{Weyl system}, a family of unitaries $W_{\xi}$ with $\xi\in\R^{2n}$ and
\begin{align*}
	W_{\xi}:=\exp(i\xi \sigma R), \quad \xi\in\R^{2n}
\end{align*}
fulfiling the Weyl relations $W_{\xi}W_{\eta}=\exp^{-i/2\xi\sigma\eta}W_{\xi+\eta}$ for all $\xi,\eta$. Such a system is unique up to isomorphism under further assumptions of continuity and irreducibility as obtained by the Stone-von Neumann theorem. Given $W_{\xi}$ it is important to note that
\begin{align}
	W_{\xi}R_kW_{\xi}^*=R_k+\xi_k\id \qquad \forall \xi\in\R^{2n}. \label{eqn:weylsystem}
\end{align}
In this paper, we will not use many properties of the Weyl system, since instead, we can work with the much simpler \emph{moments} of the state: Given a quantum state $\rho\in \mathcal{S}_1(L^2(\mathbb{R}^{2n}))$ (trace-class operators on $L^2$), its first and second centred moments are given by 
\begin{align}
	d_k&:=\tr(\rho R_k) \\
	\gamma_{kl}&:=\tr(\rho \{R_k-d_k\id,R_l-d_l\id\}_+) \label{eqn:gamma}
\end{align}
with $\{\cdot,\cdot\}_+$ the regular anticommutator. We will write $\Gamma$ instead of $\gamma$ for the covariance matrix, if we work with $\tilde{R}$ instead of $R$. Again, a simple permutation relates the two. 

An important question one can ask is when a matrix $\gamma$ can occur as a covariance matrix of a quantum state. The answer is given by Heisenberg's principle, which here takes the form of a matrix inequality:
\begin{prop}
Let $\gamma\in\mathbb{R}^{2n\times 2n}$, then there exists a quantum state $\rho$ with covariance matrix $\gamma$ if and only if 
\begin{align*}
	\gamma\geq i\sigma
\end{align*}
where $\geq$ denotes the standard partial order on matrices (i.e. $\gamma \geq i\sigma$ if $\gamma-i\sigma$ is positive semidefinite). Note that we leave out the usual factor of $\hbar/2$ to simplify notation
\end{prop}

Another question one might ask is when a covariance matrix belongs to a pure quantum state. This question cannot be answered without more information about the higher order terms. If we however require the state to be uniquely determined by its first and second moments, i.e. if we consider the so called \emph{Gaussian states}, we have an answer (cf. \cite{ade04}):
\begin{prop} \label{prop:puregaussian}
Let $\rho$ be an $n$-mode Gaussian state (i.e. completely determined by its first and second moments), then $\rho$ is pure if and only if $\det(\gamma_{\rho})=1$.
\end{prop}

\subsection{The linear symplectic group and squeezing} \label{sec:liegroups}
A very important set of operations on a quantum system are those, that leave the canonical commutation relations invariant, i.e. linear transformations $S$ such that $[SR_k,SR_l]=i\sigma_{kl}$. Such transformations are called \emph{symplectic transformations}.
\begin{dfn} Given a symplectic form $\sigma$ on $\mathbb{R}^{2n\times 2n}$, the set of matrices $S\subset \R^{2n\times 2n}$ such that $S^T\sigma S=\sigma$ is called the \emph{linear symplectic group} and is denoted by $Sp(2n,\R,\sigma)$.
\end{dfn}
We will usually drop both $\sigma$ and $\R$ in the description of the symplectic group since this will be clear from the context. The linear symplectic group is a Lie group and as such contains a lot of structure. For more information on the linear symplectic group and its connection to physics, we refer the reader to \cite{deg06} and \cite{mcd98} chapter 2. An overview for physicists is also found in \cite{arv95a}. All of the following can be found in that paper:
\begin{dfn}
Let $O(2n,\R)$ be the real orthogonal group, Then we define the following three subsets of $Sp(2n)$: 
\begin{align*}
	K(n)&:=Sp(2n,\R)\cap O(2n,\R) \\
	Z(n)&:=\left\{\id_{2(j-1)}\oplus \diag(s_i,s_i^{-1})\oplus \id_{2(n-(j+1))}|s\geq 0,j=1,\ldots,n\right\} \\
	\Pi(n)&:=\{S\in Sp(2n,\R)|S\geq 0\}
\end{align*}
The first subset is the \emph{maximally compact subgroup} of $Sp(2n)$, the second subset is the subset of \emph{single-mode-squeezers}. It generates the multiplicative subgroup $\mathcal{A}(2n)$, a maximally abelian subgroup of $Sp(2n)$. The third set is the set of positive definite symplectic matrices.
\end{dfn}

In addition, since $Sp(2n)$ is a Lie group, it possesses a \emph{Lie algebra}. Let us collect a number of relevant facts about the Lie algebra and some subsets:
\begin{prop} \label{prop:liealg}
The Lie algebra $\mathfrak{sp}(2n)$ of $Sp(2n)$ is given by
\begin{align*}
	\mathfrak{sp}(2n):=\{T\in \R^{2n\times 2n}| \sigma T+T\sigma=0\}
\end{align*}
together with the commutator as Lie bracket. Certain other Lie algebras or subsets of Lie algebras are of relevance to us:
\begin{enumerate}
	\item $\mathfrak{so}(2n):=\{A\in \R^{2n\times 2n}| A+A^T=0\}$ the Lie algebra of $SO(2n)$.
	\item $\mathfrak{k}(n):=\{A\in\R^{2n\times 2n}|A=\begin{pmatrix}{} a & b \\ -b & a\end{pmatrix}, a=-a^T, b=b^T\}$ the Lie algebra of $K(n)$. 
	\item $\pi(n):=\{A\in\R^{2n\times 2n}|A=\begin{pmatrix}{} a & b \\ b & -a\end{pmatrix}, a=a^T, b=b^T\}$ the subspace of the Lie algebra $\mathfrak{sp}(2n)$ corresponding to $\Pi(n)$.
\end{enumerate}
\end{prop}
Since the Lie algebra is a vector space, it is spanned by a set of vectors, the \emph{generators}. A standard decomposition is given by taking the generators of $\mathfrak{k}(n)$, the so called \emph{passive transformations} as one part and the generators of $\pi(n)$, the so called \emph{active transformations} as the other part. That these two sets together determine the Lie algebra completely can be seen with the \emph{polar decomposition}: 
\begin{prop}[Polar decomposition \cite{arv95a}]
For every symplectic matrix $S\in Sp(2n)$ there exists a unique $U\in K(n)$ and a unique $P\in \Pi(n)$ such that $S=UP$.
\end{prop}
A basis for the Lie algebras $\mathfrak{k}(n)$ and $\pi(n)$ therefore characterises the complete Lie algebra $\mathfrak{sp}(2n)$. Elements of the Lie algebras are also called \emph{generators} and a basis of generators therefore fixes the Lie algebra. Via the polar decomposition, this implies that they also generate the whole Lie group. We will need a set of generators $g_{ij}^{(p)}\in \mathfrak{k}(n)$ and $g_{ij}^{(a)}\in \pi(n)$ later on, which we will fix via the metaplectic representation: 
\begin{prop}[Metaplectic representation \cite{arv95a}]
Let $W_{\xi}$ be the continuous irreducible Weyl system defined above and let $S\in Sp(2n)$. Then there exists an up to a phase unique unitary $U_S$ with 
\begin{align*}
	\forall \xi:\quad U_SW_{\xi}U_S^{\dagger}=W_{S\xi}
\end{align*}
\end{prop}
Since we have the liberty of a phase, this is not really a representation of the symplectic group, but of its two-fold cover, the metaplectic group (hence the name ``metaplectic representation''). We can also study the generators of this representation, which are given by $1/2\{R_k,R_l\}_+$. For the reader familiar with annihilation and creation operators, if we denote by $a_i,a_i^{\dagger}$ the annihilation and creation operators of the $n$ bosonic modes, the generators of the metaplectic representation are given by 
\begin{align}
	G_{ij}^{p(1)}&:=i(a_j^{\dagger}a_i-a_i^{\dagger}a_j) \qquad G_{ij}^{p(2)}:=a_i^{\dagger}a_j+a_j^{\dagger}a_i \label{eqn:passivetrafo} \\
	G_{ij}^{a(3)}&:=i(a_j^{\dagger}a_i^{\dagger}-a_ia_j) \qquad G_{ij}^{a(4)}:=a_i^{\dagger}a_j^{\dagger}+a_ia_j \label{eqn:activetrafo}
\end{align}
where the $p$ stands for passive and the $a$ for active. The passive generators are also frequently called \emph{linear transformations} in the literature \cite{kok07}. We can now define a set of generators of the symplectic group $Sp(2n)$ by using the set of metaplectic generators above and take the corresponding generators in the Lie algebra $\mathfrak{sp}(2n)$ in a consistent way. As one would expect from the name, the passive metaplectic generators correspond to a set of passive generators of $\mathfrak{k}(n)$ and the set of active metaplectic generators corresponds to a set of active generators of $\pi(n)$.

With this description, we could write down the corresponding set of generators $g_{ij}$. However, we only note that the generators $G_{ii}^{a(3)}$, $i=1,\ldots,n$, correspond to the generators $g_{ii}^{a(3)}$ generating matrices in $Z_n$. This is explicitly spelled out in equations (6.6b) in \cite{arv95a}.

Given a Hamiltonian and the correspondence of the generators $G_i$ with generators of the Lie algebra $g_i$, the time evolution associated to the Hamiltonian corresponds to a path on the Lie group. Using the identifications above, we obtain the following picture: Recall that the Lie algebra $\mathfrak{g}$ of a Lie group $G$ is its tangent space $T_eG$ at the identity. The generators $g\in \mathfrak{g}$ then define left-invariant vector fields $g(\cdot):G\to TG$ with $x\mapsto g(x)\in T_xG$ by using the derivative. In matrix Lie groups, this amounts to setting $g(x)=x\cdot g$, where the product is matrix multiplication. Given a basis $g_i$ of the Lie algebra $\mathfrak{g}$ the smooth left-invariant vector fields $g_i(x)$ then define a basis of $T_xG$ at every point $x\in G$. Now, given a differentiable path $\gamma:[0,1]\to Sp(2n)$, this means that we can find differentiable coefficients $c_i:[0,1]\to \R$ such that 
\begin{align*}
	\gamma^{\prime}(t)=\sum_i c_i(t)g_i(\gamma(t))=\left(\sum_i c_i(t)g_i(e)\right)\gamma(t)=:A(t)\gamma(t)\end{align*}
where $A(t)\in \mathfrak{g}$ for all $t$ is a differentiable function. Instead of directly studying Hamiltonians with time-dependent coefficients as in equation (\ref{eqn:hamilton}), it is equivalent to study functions $A:[0,1]\to \mathfrak{g}$.

There are a number of decompositions of the Lie group and its subgroup in addition to the polar decomposition. We will mostly be concerned with the so called \emph{Euler decomposition} (sometimes called \emph{Bloch-Messiah} decomposition) and \emph{Williamson's decomposition}:
\begin{prop}[Euler decomposition \cite{arv95a}]
Let $S\in Sp(2n)$, then there exist $K,K^{\prime}\in K(n)$ and $A\in \mathcal{A}(n)$ such that $S=KAK^{\prime}$.
\end{prop}
\begin{prop}[Williamson's Theorem \cite{wil36}]
Let $M\in\mathbb{R}^{2n\times 2n}$ be a positive definite matrix, then there exists a symplectic matrix $S\in Sp(2n,\mathbb{R})$ and a diagonal matrix $D\in\mathbb{R}^{n\times n}$ such that $$M=S^T\tilde{D}S$$ where $\tilde{D}=\operatorname{diag}(D,D)$ is diagonal. The entries of $D$ are also called \emph{symplectic eigenvalues}.
\end{prop}
In particular, for $M\in \Pi(n)$, this implies that $M$ has a symplectic square root. Since covariance matrices are always positive definite, this implies also that a Gaussian state is pure if and only if its covariance matrix is symplectic. Heisenberg's uncertainty principle has also a Williamson version: 
\begin{cor} \label{cor:symplecticspec}
A positive definite matrix $M$ is a covariance matrix of a quantum state if and only if all symplectic eigenvalues are larger or equal to one.
\end{cor}
\begin{proof}
Let $M=S^T\tilde{D}S$, then $S^T\tilde{D}S\geq i\sigma~\Leftrightarrow~\tilde{D}-i\sigma\geq 0$ and one can easily check that this last inequality holds if and only if 
\begin{align*}
	\begin{pmatrix}{} d_i & i \\ -i & d_i \end{pmatrix}\geq 0 \quad \forall d_i
\end{align*}
where $d_i$ are the diagonal elements of $D$ in Williamson's Theorem. The latter however is true iff $d_i\geq 0$ for all $i$.
\end{proof}

\subsection{Quantum optical operations and squeezing}
We have already noted that an important class of operations are those, which leave the CCR-relations invariant, namely the symplectic transformations. Given a quantum state $\rho$, the action of the symplectic group on the canonical variables $R$ descends to a subgroup of unitary transformations on $\rho$ via the metaplectic representation (cf. \cite{arv95b}). Its action on the covariance matrix $\gamma_{\rho}$ of $\rho$ is even easier: Given $S\in Sp(2n)$, 
\begin{align}
	\gamma_{\rho} \mapsto S^T\gamma_{\rho}S. \label{eqn:symptraf}
\end{align}
In quantum optics, symplectic transformations can be implemented by the means of 
\begin{enumerate}
	\item beam splitters and phase shifters, implementing operations in $K(n)$ (\cite{rec94})
	\item single-mode squeezers, implementing operations in $Z(n)$.
\end{enumerate}
Via the Euler decomposition, this implies that any symplectic transformation can be implemented (approximately) by a combination of those three elements. 

\begin{dfn}
An $n$-mode bosonic state $\rho$ is called \emph{squeezed}, if its covariance matrix $\gamma_{\rho}$ possesses an eigenvalue $\lambda < 1$. 
\end{dfn}
Especially in the early literature, squeezing is usually defined differently: A state $\rho$ is squeezed if there exists a unitary transformation $K\in K(n)$ such that $K^T\gamma_{\rho}K$ has a diagonal entry smaller than one. This again comes from the physical definition of squeezed states being states where the Heisenberg uncertainty relations are satisfied with equality for at least one mode. These definitions however are well-known to be equivalent (cf. \cite{sim94}). 

\section{An operational squeezing measure for symplectic transformations} \label{sec:mathmeasure}
Throughout this section, we will always use $\sigma$ as our standard symplectic form.
\subsection{Definition and basic properties}
We will now define a first operational squeezing measure for symplectic transformations, which will later be used to define a measure for operational squeezing.
\begin{dfn}
Define the function $F:\R^{2n\times 2n}\to \R$
\begin{align}
	F(A)=\sum_{i=1}^n \log(s_i^{\downarrow}(A)) \label{eqn:defF}
\end{align}
where $s_i^{\downarrow}$ are the decreasingly ordered singular values of $A$.
\end{dfn}
Note that we sum only over half of the singular values. Restricting this function to symplectic matrices will yield an operational squeezing measure for symplectic transformations: Recall that the symplectic group is generated by symplectic orthogonal matrices and single-mode squeezers. The orthogonal matrices are easy to implement and therefore will be considered a free resource. The squeezers have singular values $s$ and $s^{-1}$ and they are experimentally hard to implement and should therefore be assigned a cost that depends on the squeezing parameter $s$. Using this, the amount of squeezing seems to be characterised by the largest singular values. Here, we quantify the amount of squeezing by a cost $\log(s)$, which can be seen as the interaction strength of the Hamiltonian needed to implement the squeezing. 

Let us make this more precise: Define the map 
\begin{align*}
	\Delta:  Sp(2n)&\to \bigcup_{m\in \N} Sp(2n)^{\times m} \\
	S &\mapsto \bigcup_{m\in\N} \{(S_1,\ldots,S_m)|S=S_1\cdots S_m, S_i\in K(n)\cup Z(n)\}
\end{align*}
The image of $\Delta$ for a given symplectic matrix contains all possible ways to construct $S$ as a product of matrices from $K(n)$ or $Z(n)$. We define:
\begin{dfn}
Let $\overline{F}:Sp(2n)\to \R$ be a map defined via
\begin{align}
	\overline{F}(S):=\log\inf\left\{ \prod_{i=1}^m s_1^{\downarrow}(S_i)\middle|(S_1,\ldots,S_m)\in \Delta(S)\right\} \label{eqn:defbarF}
\end{align}
\end{dfn}
\begin{prop} \label{prop:foverlinef}
If $S\in Sp(2n)$ then $F(S)=\overline{F}(S)$.
\end{prop}
\begin{proof}
Let $S=KAK^{\prime}$ be the Euler decomposition of $S$ with $K,K^{\prime}\in K(n)$ and $A\in\mathcal{A}(n)$. Assume without loss of generality that $A=\diag(a_1,a_1^{-1},\ldots,a_n,a_n^{-1})$ and $a_1\geq a_2 \geq \ldots \geq a_n\geq 1$ and define $A_i=\diag(1,\ldots, 1,a_i,a_i^{-1},1,\ldots,1)$. By construction $A=A_1\cdots A_n$ and $A_i\in Z(n)$. Since $K,K^{\prime}\in K(n)$, $(K,A_1,\ldots,A_n,K^{\prime})\in \Delta(S)$. Using that $s_i^{\downarrow}(K)=s_i^{\downarrow}(K^{\prime})=1$ and the fact that the Euler decomposition is actually equivalent to the singular value decomposition of $S$, we obtain:
\begin{align*}
	\overline{F}(S)&\leq \log \left(s_1^{\downarrow}(K) \prod_{i=1}^n s_1^{\downarrow}(A_i)s_1^{\downarrow}(K^{\prime})\right)
	=\log \prod_{i=1}^n s_i^{\downarrow}(S) 
	=F(S).
\end{align*}

Conversely, consider $(S_1,\ldots, S_m)\in \Delta(S)$. Using that by definition for each $S_j\in K(n)\cup Z(n)$ we have $\prod_{i=1}^n s_i^{\downarrow}(S_j)=s_1^{\downarrow}(S_j)$, we conclude:
\begin{align*}
	F(S)&=\log \left(\prod_{i=1}^n s_i^{\downarrow}(S)\right) 
		\stackrel{(*)}{\leq} \log \left(\prod_{j=1}^m\prod_{i=1}^n s_i^{\downarrow}(S_j)\right) 
		=\log \left(\prod_{j=1}^m s_1^{\downarrow}(S_j) \right)
\end{align*}
where in $(*)$ we used a special case of a theorem by Gel'fand and Naimark (\cite{bha96}, Theorem III.4.5 and equation (III.19)). Taking the infimum on the right hand side gives $F(S)\leq \overline{F}(S)$.
\end{proof}
Let us write the last observation in $(*)$ as a small lemma for later use:
\begin{lem} \label{lem:gelfand}
Let $S,S^{\prime}\in Sp(2n)$. Then $F(SS^{\prime})\leq F(S)+F(S^{\prime})$.
\end{lem}

\subsection{Lie algebraic definition} \label{sec:liealg}
Up to now, we have only considered products of symplectic matrices, which would correspond to a discrete chain of beam splitters, phase shifters and single-mode squeezers. We have also seen that the squeezing-optimal way of implementation is given by the Euler decomposition. The goal of this section is to prove that this does not change if we consider arbitrary paths on $Sp(2n)$, corresponding to general Hamiltonians of the form of equation (\ref{eqn:hamilton}) as described in section \ref{sec:preliminaries}.

Let $g^p,g^a$ denote a set of passive and active transformations for the Lie algebra as defined in equations (\ref{eqn:passivetrafo}) and (\ref{eqn:activetrafo}). We choose the normalisation of $g^a$ such that $F(\exp(cg_{ij}^a))=c$ for all $i,j$ and $c\in \R$. The normalisation of $g^p$ can be chosen arbitrarily. Furthermore, we order the generators into one vector $g^p$ and one vector $g^a$ such that $g_i^a=g_{ii}^{a(4)}$ for $i=1,\ldots,n$. Note that the $g_i^a$ generate $Z(n)$ for $i=1,\ldots,n$ as we pointed out following the equations (\ref{eqn:passivetrafo}) and (\ref{eqn:activetrafo}).

Now we need to decide which paths we want to consider. Clearly, these paths should not be pathological. We could imagine to work with continuously differentiable paths. However, the Euler decomposition will define a non-smooth path, so we have to at least allow for a finite amount of corners. 

The paths can also be defined by differential equations on the generators of the group. In other words, to each diffentiable path $\gamma:[0,1]\to Sp(2n)$ corresponds a set of differentiable coefficients $c_i:[0,1]\to \R$ such that the derivative fulfills:
\begin{align}
	\gamma^{\prime}(t)=\sum_i c_i(t)g_i(\gamma(t))=\left(\sum_i c_i(t)g_i(e)\right)\gamma(t)=:A(t)\gamma(t) \label{eqn:odepath}
\end{align}
By construction, $A(t)\in \mathfrak{g}$ is a differentiable function in $t$. Let us now consider our case $G=Sp(2n)$. Since we want the Euler decomposition to define a path, we also need to consider nondifferentiable $A$. Moreover, it makes sense physically to restrict to bounded functions $A(t)$, since $c_i(t)$ is like an interaction strength, which should not be infinite. Therefore, from a mathematical perspective, restricting $A\in L^{\infty}([0,1],\mathfrak{sp}(2n))$ seems to capture all cases we are interested in. If we make these restrictions, it turns out that the paths $\gamma$ solving differential equations of the type of equation (\ref{eqn:odepath}) are absolutely continuous. Conversely, given an absolutely continuous curve $[0,1]\to Sp(2n)$, then it is in particular of bounded variation and therefore rectifiable. Also such a curve is differentiable almost everywhere and the fundamental theorem of calculus applies, which implies that it solves a differential equation (see \cite{rud87} Theorem 7.18). Hence, the class of absolutely continuous functions with bounded derivative (almost everywhere) seems to capture all paths we are interested in.

Therefore, let $\mathcal{C}^r(S)$ be the set of absolutely continuous curves $\alpha:[0,1]\to Sp(2n)$ with a derivative which is bounded almost everywhere such that $\alpha(0)=\id$ and $\alpha(1)=S$.
\begin{dfn}
We define the function $\tilde{F}:Sp(2n)\to \R$:
\begin{align}
	\tilde{F}(S):=\inf \left\{ \int_0^1\|\vec{c}_{\alpha}^{\,a}(t)\|_1\,\mathrm{d}t\middle|\alpha\in \mathcal{C}^r(S), \dot{\alpha}(t)=(\vec{c}_{\alpha}^{\,p}(t)g^p(\alpha(t)), \vec{c}_{\alpha}^{\,a}(t)g^a(\alpha(t)))^T\right\} \label{eqn:deftildeF}
\end{align}
where we introduced the notation $\vec{c}$ to clarify that $g^{p/a}$ are actually vectors containing a set of generators each, and the coefficients might differ for each of these generators. 
\end{dfn}
The goal of this section is to prove that this does not give us any better way to avoid squeezing:
\begin{thm} \label{thm:liepaths}
For any $S\in Sp(2n)$, we have $\tilde{F}(S)=F(S)$.
\end{thm}
The proof of this theorem is quite technical and lengthy in details, thus we split it up into several lemmata. The general idea is easy to relate: A first step will be to show that for paths that descend to products of symplectic matrices of type $Z(n)$ or $K(n)$, the integral in equation (\ref{eqn:deftildeF}) is exactly equal to the corresponding product in equation (\ref{eqn:defbarF}). This is a mere consistency check (otherwise, we cannot expect the theorem to hold), but it already proves that $\tilde{F}(S)\leq F(S)$. The main part of the proof proceeds by an approximation argument: Using the differential equations defined by the paths, we can show that we can approximate the generators by paths of products of symplectic matrices to arbitrary precision, thereby proving that $\tilde{F}(S)\geq F(S)$.

In addition, we will need the following lemma:
\begin{lem} \label{lem:lieproof}
Let $A\in \mathfrak{sp}(2n)$ and write $A=1/2(A+A^T)+1/2(A-A^T)=:A_++A_-$. Then $A_+\in\pi(2n)$ and $A_-\in\mathfrak{k}(n)$ and we have $F(\exp(A))\leq F(\exp(A_+))$.
\end{lem}
\begin{proof}
First note that $F$ is continuous in $S$ since the singular values are. Using the Trotter-formula, we obtain:
\begin{align*}
	F(\exp(A))&=F\left(\lim\limits_{n\to\infty} (\exp(A_+/n)\exp(A_-/n))^n\right) \\
	&\stackrel{\mathclap{(*)}}{\leq} \lim\limits_{n\to\infty} (nF(\exp(A_+/n))+nF(\exp(A_-/n)))\\ 
	&=\lim\limits_{n\to \infty} nF(\exp(A_+/n))=F(\exp(A_+))
\end{align*}
where we used that $F(\exp(A_-))=0$ since $A_-\in\mathfrak{k}(n)$ and in $(*)$, we used a version of a theorem by Gel'fand and Naimark again (cf. \cite{bha96}, equation (III.20)).
\end{proof}

\paragraph*{First step of the proof:} Let us define yet another version of $F$ which we call $\hat{F}$ in the following way:
\begin{align*}
	C^N(S)&:=\left\{(\vec{c}_1^{\,a},\vec{c}_1^{\,p},\ldots,\vec{c}_N^{\,a},\vec{c}_N^{\,p})\middle|S=\prod_{j=1}^N \exp(\vec{c}_j^{\,a}g^a+\vec{c}_j^{\,p}g^p), \vec{c}_j\in \R^{4n^2} \right\},\\
	C(S)&:=\bigcup_{N\in\N} C^N(S), \\
	\hat{F}(S)&:=\inf\left\{ \sum_i \|\vec{c}_i^{\,a}\|_1\middle|\vec{c}\in C(S)\right\}.
\end{align*}
This definition is of course reminiscent of the definition of $\overline{F}$ in equation (\ref{eqn:defbarF}):
\begin{lem} \label{lem:fbarfhat} For $S\in Sp(2n)$, we have $\hat{F}(S)=\overline{F}(S)$.
\end{lem}
\begin{proof}
To prove $\hat{F}\leq \overline{F}$, consider the Euler decomposition $S=K_1A_1\ldots A_nK_2$ with $A_i\in Z(n)$ and $K_1,K_2\in Sp(2n)$. Since $K(n)$ is compact, the exponential map is surjective and there exist $\vec{c}_1^{\,p}$ and $\vec{c}_2^{\,p}$ such that $\exp(\vec{c}_1^{\,p}g^p)=K_1$ and $\exp(\vec{c}_2^{\,p}g^p)=K_2$. Recall that we ordered the vector $g^a$ in such a way that the generators $g^a_i$ generate the matrices in $Z(n)$ for $i=1,\ldots,n$, hence we know that there exist $\vec{c}_i^{\,a}=(0,\ldots,0, (\vec{c}_i^a)_{(i)},0,\ldots,0)$ for $i=1,\ldots,n$ such that
\begin{align*}
	S=\exp(\vec{c}_1^{\,p}g^p)\prod_{i=1}^n \exp(\vec{c}_i^{\,a}g^a)\exp(\vec{c}_2^{\,p}g^p).
\end{align*}
This implies
\begin{align*}
	\hat{F}(S)\leq \sum_i \|\vec{c}_i^{\,a}\|_1=\sum_i F(\exp(\vec{c}_i^{\,a}g^a))=\sum_i F(\exp((\vec{c}_i^{\,a})_{(i)}g^a_i))=\sum_i \log s_1(A_i)=\overline{F}(S).
\end{align*}
Here we used that $(\vec{c}_i^{\,a})_{(i)}$ is also the largest singular value of $\exp((\vec{c}_i^{\,a})_{(i)}g^a_i)\in Z(n)$, as $F(\exp((\vec{c}_i^{\,a})_{(i)}g^a_i))=(\vec{c}_i^{\,a})_{(i)}$ by normalisation of $g$.

For the other direction $\hat{F}\geq \overline{F}$, let $S$ be arbitrary. Let $c\in C(S)$ and consider each vector $\vec{c}_i$ separately. We drop the index $i$ for readability, since we need to consider the entries of the vector $\vec{c}_i$. To make the distinction clear, we denote the $j$-th entry of the vector $\vec{c}$ by $\vec{c}_{(j)}$. Recall that the active generators are exactly those generating the positive matrices. Then:
\begin{align*}
	F(\exp(\vec{c}g))\quad&\stackrel{\mathclap{\mathrm{Lemma~}\ref{lem:lieproof}}}{\leq}\quad F(\exp(\vec{c}^{\,a}g^a))
	=\lim_{n\to\infty}F\left(\left(\prod_i \exp(\vec{c}_{(i)}^{\,a}g_i^a/n)\right)^n\right) \\
	&\leq \sum_i F(\exp(\vec{c}_{(i)}^{\,a}g_i^a))=\sum_i|\vec{c}_{(i)}^a|=\|\vec{c}^{\,a}\|_1
\end{align*}
where we basically redid the calculations we used to prove Lemma \ref{lem:lieproof}, using the continuity of $F$ and the Trotter formula from matrix analysis. Until now, we have considered only one $\vec{c}_i$ of $c\in C(S)$. Now, if we define $S_i=\exp(\vec{c}_ig)$, then we have $\prod_i S_i=S$ and hence, using Lemma \ref{lem:gelfand}, we find:
\begin{align*}
	\overline{F}(S)\stackrel{\mathrm{Lemma~}\ref{lem:gelfand}}{\leq} \sum_iF(S_i)\leq \sum_iF(\exp(\vec{c}_ig)) \leq \sum_i\|\vec{c}^{\,a}_i\|_1 \qquad \forall c\in C(S)
\end{align*}
But this means $\overline{F}(S)\leq \hat{F}(S)$, as we claimed.
\end{proof}
This nearly proves $\tilde{F}(S)\leq F(S)$ for all $S\in Sp(2n)$. The remaining part is to show that $\tilde{F}$ is actually equal to $\hat{F}$ for paths corresponding to decompositions such as the Euler decomposition.

\paragraph*{Second step of the proof:} For the other direction, we need some facts about ordinary differential equations:
\begin{prop} \label{prop:odeprop}
Consider the following system of differential equations for $x:[0,1]\to \R^{2n}$:
\begin{align}
\begin{split}
	\dot{x}(t)^T&=x(t)^TA(t)\qquad \forall t\in [0,1] \\
	x(s)&=x_s \qquad \quad \mathrm{for~some~}x_s\in\mathbb{R}^{2n}, s\in [0,1] \label{eqn:ODEsystem}
\end{split}
\end{align}
where $A\in L^{\infty}([0,1],\mathfrak{sp}(2n))$. Then this system has a unique solution, which is linear in $x_s$ and defined on all of $[0,1]$ such that we can define a map
\begin{align*}
	\forall s,t\in [0,1]:\quad (s,t)\mapsto U(s,t)\in \mathcal{B}(\R^{2n})
\end{align*}
via $x(t)^T=x_s^TU(s,t)$ called the \emph{propagator} of (\ref{eqn:ODEsystem}) that fulfils:
\begin{enumerate}
	\item $U$ is continuous and differentiable almost everywhere.
	\item $U(s,\cdot)$ is absolutely continuous in $t$.
	\item $U(t,t)=\id$ and $U(s,r)U(r,t)=U(s,t)$ for all $s,t\in [0,1]$.
	\item $U(s,t)^{-1}=U(t,s)$ for all $s,t\in [0,1]$.
	\item $U$ is the unique generalised (i.e. almost everywhere) solution to the initial value problem
		\begin{align}
		\begin{split}
			\partial_t U(s,t)-U(s,t)A(t)=0 \\
			U(s,s)=\id \label{eqn:propequation}
		\end{split}
		\end{align}
		on $C([0,1]^2,\R^{2n\times 2n})$.
	\item If $A(t)=A$ does not depend on $t$, then $S(r)=\exp(rA)$ solves equation (\ref{eqn:propequation}) with $U(s,t):=S(t-s)$.
	\item for all $s,t\in[0,1]$:
	\begin{align*}
		\|U(s,t)\|_\infty \leq \exp\left(\int_{s}^t\|A(\tau)\|_1\,\mathrm{d}\tau\right).
	\end{align*}
	\item $U(s,t)\in Sp(2n)$ for all $t,s\in [0,1]$ and $\gamma(t)=U(0,t)$ fulfills equation (\ref{eqn:odepath}) with $\gamma(0)=\id$.
\end{enumerate}
\end{prop}
\begin{proof} The proof of this (except for the part about $U(s,t)\in Sp(2n)$) can be found in \cite{son98} (Theorem 55 and Lemma C.4.1) for the transposed differential equation $\dot{x}(t)=A(t)x(t)$.

For the last part, note that since $U(s,s)=\id\in Sp(2n)$, we have $U(s,s)^T J U(s,s)=J$. We can now calculate almost everywhere:
\begin{align*}
	\partial_t(U(t,s)^TJU(t,s))=-U(t,s)^T(A^T(t)J-JA(t))U(t,s)=0
\end{align*}
since $A(t)\in \mathfrak{sp}(2n)$ and therefore $A^T(t)J-JA(t)=0$. 

But this implies $U(t,s)^TJU(t,s)=J$, hence $U$ is symplectic. Obviously, $U(0,t)$ solves equation (\ref{eqn:odepath}).
\end{proof}
We will also need another well-known lemma from functional analysis:
\begin{lem} \label{lem:proofstep}
Let $A:[0,1]\to \mathfrak{sp}(2n)$, $A\in L^{\infty}([0,1],\R^{2n\times 2n})$. Then $A$ can be approximated in $\|\cdot\|_1$-norm by step-functions, which we can assume to map to $\mathfrak{sp}(2n)$ without loss of generality.
\end{lem}
The approximation by step-function can be found e.g. in \cite{rud87} (Chapter 2, exercise 24). With this in mind, we can finally prove the section's main theorem:

\begin{proof}[Proof of theorem \ref{thm:liepaths}]
$F\geq \tilde{F}$: Let $S\in Sp(2n)$ be arbitrary and consider the Euler decomposition $S=K_1S_1\cdots S_nK_2$. We can define a function $A:[0,1]\to\mathfrak{sp}(2n)$ via:
\begin{align}
	A(t):=\begin{cases} (n+2)\cdot \vec{c}_{1}^{\,p}g^p &t\in [0,1/(n+2)) \\
	 (n+2)\cdot (\vec{c}_{i+1}^{\,a})_{(i)}g^a_{i} &t\in[i/(2(n+2)),(i+1)/(n+2)),i=1,\ldots,n \\
	 (n+2)\cdot \vec{c}_{n+2}^{\,p}g^p &t\in[(n+1)/(n+2),1]
	\end{cases}
\end{align}
where $(\vec{c}_1^{\,p},0,0,\vec{c}_2^{\,a},\ldots,0,\vec{c}_{n+1}^{\,a},\vec{c}_{n+2}^{\,p},0)$ denotes the element in $C^{n+2}(S)$ for the Euler decomposition and vector indices are denoted by a subscript $_{(i)}$ as before. Let $U(s,t)$ be the propagator corresponding to $A$, then for $t\in [0,1/(n+2))$ according to Proposition \ref{prop:odeprop}, since $A$ does not depend on $t$ on this interval, it is given by $U(t,s)=\exp((t-s)A)$. In particular, $U(1/(n+2),0)=\exp(\vec{c}_{n+2}^{\,p}g^p)=K_2$. 

Iterating the procedure above, using $U(0,1)=U(0,1/(n+2))\cdots U((n+1)/(n+2),1)$, we can see that by construction, $U(0,1)=K_1S_1\cdots S_nK_2=S$. Hence $A$ defines a continuous path on $Sp(2n)$ via $U(s,t)$. We can calculate:
\begin{align*}
	\tilde{F}(S)\leq \int_0^1 \|\vec{c}^{\,a}(t)\|_1\,\mathrm{d}t &=\sum_{i=1}^n \int_{i/(n+2)}^{(i+1)/(n+2)} |(n+2)\cdot (\vec{c}^{\,a}_{i+1})_{(i)}|\,\mathrm{d}t \\
	&=\sum_{i=1}^n |(\vec{c}^{\,a}_{i+1})_{(i)}|\stackrel{\mathrm{Lemma}~\ref{lem:fbarfhat}}{=}F(S)
\end{align*}
where we used that the integral over the interval $[0,1/(n+2))$ and $[(n+1)/(n+2),1]$ is empty due to the fact that all active components are zero. In the last step, we used that for the Euler decomposition, which takes the minimum in $\hat{F}$, this value is exactly $\sum_i |(\vec{c}^{\,a}_{i+1})_{(i)}|=\sum_i \|\vec{c}^{\,a}_{i+1}\|_1$, since $(\vec{c}_{i+1}^a)_{(j)}=0$ for $j\neq i$. 

$F\leq \tilde{F}$: Let $S\in Sp(2n)$ be arbitrary. While proving $F\geq \tilde{F}$, we have seen for a special example that propagators for step-functions $A:[0,1]\to\mathfrak{sp}(2n)$ relate to vectors $c\in C^N(S)$, where $N$ is the number of steps. We have also seen that the measure of $c$ and the measure of the path $\alpha$ on $Sp(2n)$ corresponding to $A$ in the definition of $\hat{F}$ and $\tilde{F}$ are equivalent in this case. 

This means that using Proposition \ref{prop:foverlinef} and Lemma \ref{lem:fbarfhat} we have already proved:
\begin{align*}
	F(S)=\inf &\left\{ \int_0^1\|\vec{c}_{\alpha}^{\,a}(t)\|_1\,\mathrm{d}t\middle|\alpha\in \mathcal{C}^r(S),\right. \\
		&\left. \dot{\alpha}(t)=(\vec{c}_{\alpha}^{\,p}(t)g^p(\alpha(t)), \vec{c}_{\alpha}^{\,a}(t)g^a(\alpha(t)))^T, \vec{c}\mathrm{~step~fct.}\right\}.
\end{align*}
The only thing left to prove is that we can drop the step-function assumption. This will be done by an approximation argument: Let $\varepsilon>0$ be arbitrary and let $A\in L^{\infty}([0,1],\R^{2n\times 2n})$ with $A(t)\in \mathfrak{sp}(2n)$ always. Then $A$ defines a path $\alpha:[0,1]\to Sp(2n)$ via Proposition \ref{prop:odeprop}, which in turn defines $c_{\alpha}(t)g(\alpha(t))=\dot{\alpha}(t)$ at any point with the usual generators $g$ of $\mathfrak{sp}(2n)$. We know that $c_{\alpha}:[0,1]\to \R^{4n^2}$ is an $L^{\infty}$-function, since $A \in L^{\infty}$. Now consider an arbitrary $A\in L^{\infty}$ such that
\begin{align}
	\left| \int_{0}^1 \|\vec{c}^{\,a}_{\alpha}(t)\|_1\,\mathrm{d}t-\tilde{F}(S)\right|<\varepsilon \label{eqn:bigproof1}
\end{align}
i.e. $A$ corresponds to a path that is close to the infimum in the definition of $\tilde{F}$. We can now approximate $c_{\alpha}$ by step-functions $c_{\alpha^{\prime}}$ (corresponding to a function $A^{\prime}$, see Lemma \ref{lem:proofstep}) such that
\begin{align}
	\left| \int_0^1 \|c_{\alpha}(t)-c_{\alpha^{\prime}}\|_1\,\mathrm{d}t\right|<\varepsilon \label{eqn:bigproof2}
\end{align}
Using the fact that the propagators $U_A,U_{A^{\prime}}$ are differentiable almost everywhere (Proposition \ref{prop:odeprop}) and absolutely continuous when one entry is fixed, we can define a function $f(s):=U_A(0,s)U_{A^{\prime}}(s,t)$, which is also differentiable almost everywhere. Furthermore, the fundamental theorem of calculus holds for $f(s)$, since it is absolutely continuous as $U(s,t)$ is absolutely continuous in $s$ or $t$ and the fundamental theorem holds for absolutely continuous functions (cf \cite{rud87}, theorems 6.10 and 7.8).
\begin{align*}
	\frac{\mathrm{d}}{\mathrm{d}s}f(s)=-U_A(0,s)A(s)U_{A^{\prime}}(s,t)+U_A(0,s)A^{\prime}(s)U_{A^{\prime}}(s,t)
\end{align*}
almost everywhere, which implies:
\begin{align*}
	U_{A^{\prime}}(0,t)-U_A(0,t)&=f(t)-f(0)=\int_0^t \frac{\mathrm{d}}{\mathrm{d}s}f(s)\,\mathrm{d}s \\
	&=\int_0^t U_A(0,s)(A^{\prime}(s)-A(s))U_{A^{\prime}}(s,t)\,\mathrm{d}s
\end{align*}
and:
\begin{align}
	\|U_{A^{\prime}}(0,t)-U_A(0,t)\|_1&\leq \int_0^t \|A^{\prime}(s)-A(s)\|_{1}\,\mathrm{d}s \sup_{\tau\in [0,t]}\|U_{A^{\prime}}(0,\tau)\|_{\infty} \sup_{\tau\in [0,t]}\|U_A(0,\tau)\|_{\infty} \nonumber\\
	&\leq M\cdot \int_0^t \|A^{\prime}(s)-A(s)\|_{1}\,\mathrm{d}s \nonumber\\
	&= M\cdot \int_0^t \|c_{\alpha^{\prime}}(s)g-c_{\alpha}(s)g\|_{1}\,\mathrm{d}s \nonumber\\
	&\leq M\cdot \|g\|_{\infty} \int_0^t \|c_{\alpha^{\prime}}(s)-c_A(s)\|_{1}\,\mathrm{d}s \nonumber \\
	&\leq M M^{\prime} \varepsilon \label{eqn:proof1}
\end{align}
for some $M<\infty, M^{\prime}<\infty$ using that the propagators are bounded ($M$ can explicitly be computed by the bounds given in Proposition \ref{prop:odeprop}). 

Up to now, we have taken a path $\alpha$ to $S$ close to the infimum and approximated it by a path $\alpha^{\prime}$. It is immediate by equations (\ref{eqn:bigproof1}) and (\ref{eqn:bigproof2}) that
\begin{align}
	\left| \int_0^1 \| c_{\alpha^{\prime}}(t)\|_1\,\mathrm{d}t-\tilde{F}(S)\right|<2\varepsilon. \label{eqn:proof4}
\end{align}
Since $c_{\alpha^{\prime}}\in C^N(S^{\prime})$ for some $N\in\N$ and $S^{\prime}=U_{A^{\prime}}(0,1)$, we would be done if $S^{\prime}=S$. To remedy this, we want to extend $\alpha^{\prime}$ to a path $\tilde{\alpha}$ such that it ends at $S$. This is where equation (\ref{eqn:proof1}) enters: Set $\tilde{S}:=U_{A^{\prime}}(0,1)^{-1}U_A(0,1)$, then
\begin{align}
	\|\tilde{S}-\id\|_1\leq \frac{MM^{\prime}}{\|S\|_1}\varepsilon \label{eqn:proof2}
\end{align}
hence $\tilde{S}\approx \id$ for $\varepsilon$ small enough. Using the polar decomposition, we can write $\tilde{S}=\exp(\vec{c}_{N+1}^{\,p})\exp(\vec{c}_{N+2}^{\,a})$. From equation (\ref{eqn:proof2}) we obtain
\begin{align*}
	\|\tilde{S}\|_1<n+\varepsilon n\frac{MM^{\prime}}{\|S\|_1}<\|\tilde{S}\|_1+2n\varepsilon \frac{MM^{\prime}}{\|S\|_1}, \\
	\|\tilde{S}\|_{\infty}<1+\varepsilon \frac{MM^{\prime}}{\|S\|_1},
\end{align*}
which in turn implies for each singular value $s_i(\tilde{S})\leq \varepsilon \frac{MM^{\prime}}{\|S\|_1}$ for $i=1,\ldots,n$ and therefore 
\begin{align}
	\|\log \tilde{S}\|_1\leq n\log\left(\varepsilon \frac{MM^{\prime}}{\|S\|_1}\right)\leq n\varepsilon \frac{MM^{\prime}}{\|S\|_1}=:C\varepsilon. \label{eqn:proof3} 
\end{align}

This lets us construct a new  $\tilde{A}: [0,2]\to \mathfrak{sp}(2n)$:
\begin{align*}
	t\mapsto \begin{cases} A^{\prime}(t) &t\in [0,1] \\
		2\cdot\vec{c}_{N+1}^{\,p}g^p &t\in(1,3/2) \\
		2\cdot\vec{c}_{N+2}^{\,a}g^a &t\in(3/2,2] 
		\end{cases}.
\end{align*}
By construction, for the corresponding propagator we have $U_{\tilde{A}}(0,2)=S$ and $\tilde{\alpha}$ is a feasible path for $\tilde{F}(S)$ (at least after reparameterisation) fulfiling:
\begin{align*}
	\left|\int_0^2 \|c_{\tilde{\alpha}}^a(t)\|_1\,\mathrm{d}t-\tilde{F}(S)\right|
	&\leq \left|\int_0^1 \|c_{\tilde{\alpha}}^a(t)\|_1\,\mathrm{d}t-\tilde{F}(S)\right|+
	\left|\int_0^1 \|c_{\tilde{\alpha}}^a(t)\|_1\,\mathrm{d}t\right| \\
	&\stackrel{\mathclap{(\ref{eqn:proof4})+(\ref{eqn:proof3})}}{\leq}\quad (2+C)\varepsilon
\end{align*}
Since, $c_{\tilde{\alpha}}\in C^{N+2}(S)$, $\tilde{\alpha}$ is a valid path for $\hat{F}(S)$, which implies that for any $\epsilon>0$, choosing $\varepsilon:=\epsilon/(2+C)$, we have seen:
\begin{align}
	\hat{F}(S)<\tilde{F}(S)+\epsilon
\end{align}
For $\epsilon\to 0$, $\hat{F}(S)\leq \tilde{F}(S)$, which implies $F(S)\leq \tilde{F}(S)$ via Lemma \ref{lem:fbarfhat}.
\end{proof}

\section{A mathematical measure for squeezing of arbitrary states} \label{sec:measureg}
Throughout this section, for convenience, we will switch to using $J$ as symplectic form. Having defined the measure $F$, we will now proceed to define a squeezing measure for creating an arbitrary (mixed) state:
\begin{dfn}
Let $\rho$ be an $n$-mode bosonic quantum state with covariance matrix $\Gamma$. We then define:
\begin{align}
	G(\rho)\equiv G(\Gamma):=\inf\{F(S)|\Gamma\geq S^{T\!}S, S\in Sp(2n)\} \label{eqn:defG}
\end{align}
\end{dfn}
How is this motivated? The vacuum has covariance matrix $\Gamma_{\mathrm{vac}}=\id$ and Gaussian noise is given by covariance matrices $E\geq 0$, thus any covariance matrix $\Gamma_0\geq \id$ should be free to produce. Given $\Gamma$, for any $S^TS\leq \Gamma$, we can find $\Gamma_0\geq \id$ such that $S\Gamma_0S^T=\Gamma$ and we have then created a state with covariance matrix $\Gamma$ where the only squeezing costs required are those required for $S$. Since common wisdom from older measures (cf. \cite{kra03b}) tells us that measurements, ancillas, or convex combinations cannot increase squeezing, $G$ seems a likely candidate for an operational measure. $G$ is also reminiscent of already existing measures for Gaussian states such as the entanglement of formation \cite{wol04}.

Note that $G$ is always finite: For any given covariance matrix $\Gamma$, by Williamson's Theorem and Corollary \ref{cor:symplecticspec}, we can find $S\in Sp(2n)$ and $\tilde{D}\geq \id$ such that $\Gamma=S^T\tilde{D}S\geq S^TS$. Furthermore $G$ is also nonnegative since $F$ is nonnegative for symplectic $S$, as for any symplectic matrix $S$ the $n$ biggest singular values $s_i$ need to fulfil $s_i\geq 1$. 

\subsection{Different reformulations of the measure}
We will now give several reformulations of the squeezing measure and prove some of its properties. In particular, $G$ is convex and one of the crucial steps towards proving convexity of $G$ is given by a reformulation of $G$ with the help of the Cayley transform:
\begin{prop} \label{prop:cayley}
Define the Cayley transform and its inverse via:
\begin{align}
\begin{split}
	\Ca:\{H\in\R^{n\times n}|\operatorname{spec}(H)\cap\{+1\}=\emptyset\}\to \R^{n\times n}\\
	H \mapsto \frac{\id+H}{\id-H}
\end{split} \\
\begin{split}
	\Ca^{-1}\{S\in\R^{n\times n}|\operatorname{spec}(H)\cap\{-1\}=\emptyset\}\to\R^{n\times n} \\
	S\mapsto \frac{S-\id}{S+\id}
\end{split}
\end{align}
$\Ca$ is a diffeomorphism onto its image with inverse $\Ca^{-1}$. Furthermore, it has the following properties:
\begin{enumerate}
	\item $\mathcal{C}$ is operator monotone and operator convex on matrices $A$ with $\operatorname{spec}(A)\subset (-1,1)$.
	\item $\mathcal{C}^{-1}$ is operator monotone and operator concave on matrices $A$ with $\operatorname{spec}(A)\subset (-1,\infty)$.
	\item $\Ca:\mathbb{R}\to \mathbb{R}$ with $\Ca(x)=(1+x)/(1-x)$ is log-convex on $[0,1)$. 
	\item For $n=2m$ even, $H\in \R^{2m\times 2m}$ and $H\in \mathcal{H}$ if and only if $\Ca(H)\in Sp(2m,\R)$ and $\Ca(H)\geq iJ$.
\end{enumerate}
where $\mathcal{H}$ is defined via:
\begin{align}
	\mathcal{H}=\left\{H=\begin{pmatrix}{} A & B \\ B & -A\end{pmatrix}\in\R^{2m\times 2m}\middle| A^T=A, B^T=B, \operatorname{spec}(H)\subset (-1,1)\right\}
\end{align}
\end{prop}
The definition and the fact that this maps the upper half plane of positive definite matrices to matrices inside the unit circle is present in \cite{arn88} (I.4.2) and \cite{mcd98} (Prop. 2.51, Proof 2). Since no proof is given in the references and they do not cover the whole proposition, we provide one in Appendix \ref{app:cay}. We can use the Cayley transform to reformulate our squeezing measure:
\begin{prop} \label{prop:reform}
Let $\Gamma\geq iJ$ and $\Gamma\in\R^{2n\times 2n}$ symmetric. Then:
\begin{align}
	G(\Gamma)&=\inf\{F(S)|\Gamma\geq S^TS, S\in Sp(2n)\} \label{eqn:reform1}\\
	&=\inf\{F(\Gamma_0^{1/2})|\Gamma\geq \Gamma_0\geq iJ\} \label{eqn:reform2}\\
	&=\inf\left\{\frac{1}{2}\sum_{i=1}^n \log \left(\frac{1+s_i(A+iB)}{1-s_i(A+iB)}\right)\middle|\Ca^{-1}(\Gamma)\geq H, H\in \mathcal{H}\right\} \label{eqn:reform3}
\end{align}
\end{prop}
\begin{proof} First note that the infimum in all three expressions is actually attained. We can see this most easily in the first equation (\ref{eqn:reform1}): The matrix inequalities $\Gamma\geq S^TS(\geq iJ)$ imply that the set of feasible $S$ in the minimization is compact, hence its minimum is attained. For the other two equations, this follows similarly.

To see (\ref{eqn:reform1}) $=$ (\ref{eqn:reform2}), first note that (\ref{eqn:reform2}) $\leq$ (\ref{eqn:reform1}) since any $S\in Sp(2n)$ also fulfils $S^TS\geq iJ$, hence $\Gamma\geq S^TS\geq iJ$. For equality, note that for any $\Gamma\geq \Gamma_0\geq iJ$, using Williamson's Theorem we can find $S\in Sp(2n)$ and a diagonal $\tilde{D}=(d_1,\ldots,d_n,d_1,\ldots,d_n)$ with $d_i\geq 1$ for all $i=1,\ldots n$ such that $\Gamma_0=S^TDS\geq S^TS\geq iJ$ via Corollary \ref{cor:symplecticspec}. But since $F\left(\Gamma_0^{1/2}\right)\geq F((S^TS)^{1/2})=F(S)$ via the Weyl monotonicity principle, we know that the infimum is achieved on symplectic matrices.

Finally, let us prove equality with (\ref{eqn:reform3}). Observe that as $S^TS\geq iJ$ and $S^TS$ is symplectic, $S^TS=\Ca(H)$ for some $H\in\mathcal{H}$ and conversely, for any $H\in\mathcal{H}$ there exists a symplectic matrix $S$ with $\Ca(H)=S^TS$ and we can replace $Sp(2n)$ by $\mathcal{H}$ using the Cayley transform. 

Using the fact that $s_i^{\downarrow}(S)=\lambda_i^{\downarrow}(S^TS)^{1/2}=\lambda_i^{\downarrow}(\Ca(H))^{1/2}$ and the fact that $H$ is diagonalised by the same unitary matrices as $\Ca(H)=(\id+H)\cdot(\id-H)^{-1}$ whence its eigenvalues are 
\begin{align*}
	\lambda_i^{\downarrow}(\Ca(H))=\frac{1+\lambda_i^{\downarrow}(H)}{1-\lambda_i^{\downarrow}(H)},
\end{align*}
we have:
\begin{align*}
	\inf\{F(S)|\Gamma\geq S^TS, S\in Sp(2n)\}=\inf\left\{\log \prod_{i=1}^n \left( \frac{1+\lambda^{\downarrow}_i(H)}{1-\lambda_i^{\downarrow}(H)}\right)^{\frac{1}{2}}\middle|\Gamma\geq \Ca(H), H\in \mathcal{H}\right\}
\end{align*}
Next we claim $\lambda^{\downarrow}_i(H)=s_i^{\downarrow}(A+iB)$ for $i=1,\ldots,n$. To see this note:
\begin{align}
	\frac{1}{2} \begin{pmatrix}{}	\id & i\id \\ \id & -i\id\end{pmatrix}
	\cdot \begin{pmatrix}{} A & B \\ B & -A\end{pmatrix}
	\cdot \begin{pmatrix}{}	\id & \id \\ -i\id & i\id\end{pmatrix}
	=\begin{pmatrix}{} 0 & A+iB \\ A-iB & 0\end{pmatrix} \label{eqn:aib}
\end{align}
Since we conjugated with a unitary matrix, the eigenvalues of this matrix are the same as the eigenvalues of $H$. The singular values of the matrix on the right hand side of equation (\ref{eqn:aib}), which are the absolute values of the eigenvalues of $H$, are the eigenvalues of $\diag( (A+iB)^{\dagger}(A+iB), (A+iB)(A+iB)^{\dagger})^{1/2}$, which are the singular values of $A+iB$ with double multiplicity. From the structure of $H$, it is immediate that the eigenvalues of the right hand side of equation (\ref{eqn:aib}) and thus of $H$ come in pairs $\pm s_i(A+iB)$. Hence $\lambda_i^{\downarrow}(H)=s_i^{\downarrow}(A+iB)$ for $i=1,\ldots,n$ and we have:
\begin{align*}
	\inf&\{F(S)|\Gamma\geq S^TS, S\in Sp(2n)\} \\&=\inf\left\{\frac{1}{2}\sum_{i=1}^n \log \left(\frac{1+s_i(A+iB)}{1-s_i(A+iB)}\right)\middle|\Gamma\geq \mathcal{C}(H), H\in \mathcal{H}\right\}
\end{align*}
To see that that the right hand side equals (\ref{eqn:reform3}), we only need to use the fact that $\Gamma \geq \mathcal{C}(H) \Leftrightarrow \mathcal{C}^{-1}(\Gamma)\geq H$ for all $H\in \mathcal{H}$ and $\Gamma\geq iJ$ since the Cayley transform and its inverse are operator monotone.
\end{proof}

\subsection{Convexity}
The reformulation (\ref{eqn:reform3}) will allow us to prove:
\begin{thm} \label{thm:convex} $G$ is convex on the set of covariance matrices $\{\Gamma\in \R^{2n\times 2n}| \Gamma\geq iJ\}$.
\end{thm}
Before we start with the proof, let us recall two lemmata from matrix analysis:
\begin{lem}[Lidskii's Theorem] \label{lem:lidskii}
Let $A,B\in\C^{n\times n}$ be hermitian, then $\lambda^{\downarrow}(A+B)$ lies in the convex hull of $\lambda^{\downarrow}(A)+P_{\pi}\lambda^{\downarrow}(B)$, where $\lambda^{\downarrow}$ denotes the vector of eigenvalues in decreasing order and $\pi\in S_n$ is an arbitrary permutation with its permutation matrix $P_\pi$. 

Put differently:
\begin{align}
	\lambda_i^{\downarrow} (A+B)=\sum_{\pi \in S_n} p_\pi(\lambda_i^{\downarrow}(A)+\lambda_{\pi(i)}^{\downarrow}(B)) \qquad~~ p_{\pi}\geq 0, \sum_\pi p_\pi=1 \label{eqn:lidskii}
\end{align}
\end{lem}
A proof of this theorem can be found in \cite{bha96}, the explicit formulation is given in Exercise III.4.3.
\begin{lem}[Thompson's Theorem, \cite{tho76}] \label{lem:thompson}
Let $A,B\in\C^{n\times n}$, then there exist unitaries $U,V$ such that
\begin{align}
	|A+B|\leq U|A|U^*+V|B|V^* \label{eqn:thompson}
\end{align}
where $|A|=(A^*A)^{1/2}$ denotes the absolute value.
\end{lem}
Let us prove another lemma that will be the crucial part of the proof of convexity:
\begin{lem} \label{lem:convf}
Consider the map $f:\mathbb{R}^{n\times n}\times \mathbb{R}^{n\times n}\to \mathbb{R}$:
\begin{align}
	f(A,B)=\frac{1}{2}\sum_{i=1}^n \log \left(\frac{1+s_i(A+iB)}{1-s_i(A+iB)}\right) \label{eqn:convexsum}
\end{align}
If we restrict $f$ to symmetric matrices $A$ and $B$ such that $s_i(A+iB)<1$ for all $i=1,\ldots,n$, $f$ is jointly convex in $A,B$, i.e.
\begin{align*}
	f(t A+(1-t) A^{\prime},t B+(1-t) B^{\prime})\leq 
		t f(A,B)+(1-t) f(A^{\prime},B^{\prime}) \qquad \forall \,t\in [0,1]
\end{align*}
\end{lem}
\begin{proof}
Let $\tilde{A}:=t A+(1-t) A^{\prime}$ and $\tilde{B}:=t B+(1-t) B^{\prime}$. Note that $\tilde{A}$ and $\tilde{B}$ are also symmetric, and the largest singular value of $\tilde{A}+i\tilde{B}$ fulfils $s_1^{\downarrow}(\tilde{A}+i\tilde{B})\leq t s_1^{\downarrow}(A+iB)+(1-t)s_1^{\downarrow}(A^{\prime}+iB^{\prime})$. Therefore, the singular values of any convex combination of $A+iB$ and $A^{\prime}+iB^{\prime}$ also lie in the interval $[0,1)$. This makes our restriction well-defined under convex combinations.

Then for any $j=1,\ldots,n$:
\begin{align}
	s_j(\tilde{A}+i\tilde{B})&=\lambda_j(|\tilde{A}+i\tilde{B}|)\quad \stackrel{\mathclap{\mathrm{Lemma}~\ref{lem:thompson}}}{\leq}\quad \lambda_j(U|t (A+iB)|U^{*}+V|(1-t)(A^{\prime}+iB^{\prime})|V^{*}) \nonumber \\
	&\stackrel{\mathclap{\mathrm{Lemma}~\ref{lem:lidskii}}}{\leq}\quad \lambda_j(U|t(A+iB)|U^{*})+\sum_{\pi} p_{\pi} \lambda_{\pi(j)}(V|(1-t)(A^{\prime}+iB^{\prime})|V^{*}) \nonumber \\
	&\stackrel{\mathclap{(*)}}{=}\lambda_j(|t(A+iB)|)+\sum_{\pi} p_{\pi} \lambda_{\pi(j)}(|(1-t)(A^{\prime}+iB^{\prime})|) \nonumber \\
	&=t\lambda_j(|A+iB|)+(1-t)\sum_{\pi} p_{\pi} \lambda_{\pi(j)}(|A^{\prime}+iB^{\prime}|)
\end{align}
with $p_\pi\geq 0$ and $\sum_\pi p_\pi=1$. In $(*)$, we used that unitaries do not change the spectrum.

Now each summand in equation (\ref{eqn:convexsum}) is the Cayley transform of a singular value. We can use the log-convexity of the Cayley-transform to prove the joint convexity of $f$:
\begin{align*}
	f(\tilde{A},\tilde{B})&=\sum_{i=1}^n \log \Ca[s_i(\tilde{A}+i\tilde{B})] \\
	&\leq\sum_{i=1}^n \log \Ca\left[t\lambda_i(|A+iB|)+(1-t)\sum_{\pi} p_{\pi} \lambda_{\pi(i)}(|A^{\prime}+iB^{\prime}|)\right] \\
	&\leq \sum_{i=1}^n \left(t \log \Ca[\lambda_i(|A+iB|)]+(1-t)\sum_{\pi}p_{\pi}\log\Ca[\lambda_{\pi(i)}(|A^{\prime}+iB^{\prime}|)]\right) \\
	&= \sum_{i=1}^n t \log \Ca[\lambda_i(|A+iB|)]+(1-t)\sum_{\pi}p_{\pi}\left( \sum_{i=1}^n\log \Ca[\lambda_{\pi(i)}(|A^{\prime}+iB^{\prime}|)]\right) \\
	&\leq t \sum_{i=1}^n \log \Ca[\lambda_i(|A+iB|)]+(1-t)\sum_{\pi}p_{\pi}\cdot \max_{\pi}\left(\sum_{i=1}^n \log \Ca[\lambda_{\pi(i)}(|A^{\prime}+iB^{\prime}|)]\right) \\
	&\stackrel{\mathclap{(**)}}{=} ~t \sum_{i=1}^n \log \Ca[\lambda_i(|A+iB|)]+(1-t)\sum_{i=1}^n \log \Ca[\lambda_i(|A^{\prime}+iB^{\prime}|)] \\
	&=tf(A,B)+(1-t)f(A^{\prime},B^{\prime})
\end{align*}
where in $(**)$ we use that the sum of all eigenvalues is of course not dependent on the order of the eigenvalues. 
\end{proof}
This lemma will later allow us to calculate $G$ as a convex program.

\begin{proof}[Proof of Theorem \ref{thm:convex}] We can now finish the proof of the convexity of $G$.

First note that using the definition of $f$ in Lemma \ref{lem:convf} we can reformulate (\ref{eqn:reform3}) to
\begin{align}
	G(\Gamma)=\inf\left\{f(A,B)\middle| \mathcal{C}^{-1}(\Gamma)\geq H, H\in \mathcal{H}\right\}. \label{eqn:reform4}
\end{align}

Let $\Gamma\geq iJ, \Gamma^{\prime}\geq iJ$ be two covariance matrices and let $H,H^{\prime}\in\mathcal{H}$ be the matrices that attain the minimum of $G(\Gamma),G(\Gamma^{\prime})$ respectively. Then, in particular, $tH+(1-t)H^{\prime}\in\mathcal{H}$. Furthermore, since $\Ca^{-1}(\Gamma)\geq H$ and $\Ca^{-1}(\Gamma^{\prime})\geq H^{\prime}$ we have
\begin{align*}
	\Ca^{-1}(t \Gamma+(1-t)\Gamma^{\prime})\stackrel{(*)}{\geq} t\Ca^{-1}(\Gamma)+(1-t)\Ca^{-1}(\Gamma^{\prime})\geq tH+(1-t)H^{\prime}
\end{align*}
where we used the operator concavity of $\Ca^{-1}$ in $(*)$. This means that $tH+(1-t)H^{\prime}$ is a feasible matrix for the minimisation in $G$, which implies using equation (\ref{eqn:reform4})
\begin{align*}
	G(t\Gamma+(1-t)\Gamma^{\prime})\leq f(tA+(1-t)A^{\prime},tB+(1-t)B^{\prime}).
\end{align*}
The convexity now follows directly from Lemma \ref{lem:convf} and the fact that we chose $H$ and $H^{\prime}$ to attain $G(\Gamma)$ and $G(\Gamma^{\prime})$.
\end{proof}

\subsection{Continuity properties}
From the convexity of $G$ on the set of covariance matrices, it follows from general arguments in convex analysis that $G$ is continuous on the interior of the set of covariance matrices (cf. \cite{roc97}, theorem 10.1). What more can we say about the boundary?

\begin{thm} \label{thm:continuous} 
$G$ is lower semicontinuous on the set of covariance matrices $\{\Gamma\in \R^{2n\times 2n}| \Gamma\geq iJ\}$ and continuous on its interior. Moreover, $G(\Gamma+\varepsilon \id)\to G(\Gamma)$ for $0<\varepsilon\to 0$ for any $\Gamma\geq iJ$.
\end{thm}
In Corollary \ref{cor:symplecticspec}, we have seen that the set of covariance matrices is the subset of positive definite matrices with symplectic spectrum contained in $[1,\infty)$. This implies that the boundary of this set is given by all positive definite covariance matrices with at least one symplectic eigenvalue equal to one.

The ultimate goal is to extend the continuity from the interior to the exterior. At present we are only able to prove lower semicontinuity and to give a sequence converging from the interior. The proof will need a few notions from set-valued analysis that we review in appendix \ref{app:setvalue}.

\begin{proof}[Proof of Theorem \ref{thm:continuous}]
As already observed, $G$ is continuous on the interior. Let $\Gamma_0\geq iJ$ be arbitrary and suppose
\begin{align*}
	\mathcal{A}(\Gamma):=\{\hat{\Gamma}|(\Gamma-2\|\Gamma_0\|\id)\leq \hat{\Gamma}\leq \Gamma\}.
\end{align*}
By definition, $\mathcal{A}$ is compact and convex for any $\Gamma$. Moreover, it defines a set-valued function on the set of covariance matrices with nonempty values. Let $\varepsilon>0$, then for all $\Gamma\geq iJ$ with $\|\Gamma-\Gamma_0\| <\varepsilon$, we have that for any $\hat{\Gamma}\in \mathcal{A}(\Gamma)$, $\tilde{\Gamma}:=\hat{\Gamma}+(\Gamma-\Gamma_0) \in \mathcal{A}(\Gamma_0)$ and $\|\hat{\Gamma}-\tilde{\Gamma}\|<\varepsilon$. This is the condition in Lemma \ref{lem:uHsc} hence the set-valued function defined by $\mathcal{A}$ is upper semicontinuous at $\Gamma_0$, which implies that $\mathcal{A}(\Gamma)\cap \{X|iJ\leq X\}$ is also upper semicontinuous by Proposition \ref{prop:usc}. If $\varepsilon$ is small enough (e.g. $\varepsilon<1$), this implies
\begin{align*}
	\mathcal{A}(\Gamma)\cap \{X|iJ\leq X\} = \{X|iJ\leq X\leq \Gamma\}=:\mathcal{G}(\Gamma),
\end{align*}
hence this set is upper semicontinuous at $\Gamma_0$. 

Since $F$ is continuous on positive definite matrices, it is absolutely continuous if we restrict to a small neighbourhood of the covariance matrix $\Gamma_0$. More precisely, let us restrict $F$ to $\bigcup_{\|\Gamma-\Gamma_0\|<1} \mathcal{G}(\Gamma)$. This means that for every $\varepsilon>0$ there exists an $\epsilon>0$ such that 
\begin{align}
	F(\tilde{\Gamma})-\varepsilon<F(\hat{\Gamma})<F(\tilde{\Gamma})+\varepsilon \label{eqn:contproof1}
\end{align}
for all $\|\tilde{\Gamma}-\hat{\Gamma}\|<\epsilon$ and all $\tilde{\Gamma},\hat{\Gamma}\in\bigcup_{\|\Gamma-\Gamma_0\|<1} \mathcal{G}(\Gamma)$. 

Assuming without loss of generality that $\|\Gamma-\Gamma_0\|<1$, the set $\mathcal{G}(\Gamma)$ is exactly the set for the minimisation in the definition of $G$. The upper semicontinuity of $\mathcal{G}(\Gamma)$ implies by Lemma \ref{lem:uHsc} that for every $\epsilon>0$ there exists a $\delta>0$ such that for all $\|\Gamma-\Gamma_0\|<\delta$ we have: For all $\hat{\Gamma}\in \mathcal{G}(\Gamma)$ there exists a $\tilde{\Gamma}\in \mathcal{G}(\Gamma_0)$ such that $\|\hat{\Gamma}-\tilde{\Gamma}\|<\epsilon$. In particular, this is true for all minimisers $\hat{\Gamma}$ with $G(\Gamma)=F(\hat{\Gamma}^{1/2})$. Both $\hat{\Gamma}$ and $\tilde{\Gamma}\in \bigcup_{\|\Gamma-\Gamma_0\|<1} \mathcal{G}(\Gamma)$. Using equation (\ref{eqn:contproof1}) we obtain: for every $\varepsilon>0$ there exists a $\delta>0$ such that for all $\|\Gamma-\Gamma_0\|<\delta$, we have a pair $\hat{\Gamma},\tilde{\Gamma}$ with $\hat{\Gamma}\in \mathcal{G}(\Gamma)$ minimising $G(\Gamma)$ and $\tilde{\Gamma}\in \mathcal{G}(\Gamma_0)$ such that
\begin{align*}
	F(\tilde{\Gamma})-\varepsilon<F(\hat{\Gamma})=G(\Gamma)
\end{align*}
This implies that for all $\varepsilon>0$ there exists a $\delta>0$ such that
\begin{align*}
	G(\Gamma_0)\leq G(\Gamma)+\varepsilon
\end{align*}
for all $\|\Gamma-\Gamma_0\|<\delta$. 

In particular, we can also take the limit inferior on both sides to obtain
\begin{align*}
	G(\Gamma_0)\leq \liminf_{\Gamma\to \Gamma_0} G(\Gamma)
\end{align*}
which means that $G$ is lower semicontinuous at $\Gamma_0$. 

If one proved that the set $\mathcal{G}(\Gamma_0)$ is also lower semicontinuous, a similar argument would prove that $G$ would be upper semicontinuous at $\Gamma_0$. 

Finally, let us prove that $G(\Gamma_0+\varepsilon \id)\to 0$ for $\varepsilon\to 0$. To see this, consider the closed sets
\begin{align*}
	C_{n}:=\bigcup_{0\leq \xi\leq 1/n} \mathcal{G}(\Gamma_0+\xi \id)
\end{align*}
for any $n\in \N$. It is easy to see that $C_{n+1}\subseteq C_{n}$ and that $\bigcap_{n\in \infty} C_n = \mathcal{G}(\Gamma_0)$. Moreover, $C_1$ is compact. Now let $\Gamma_n$ be the sequence of minimisers for $G(\Gamma_0+1/n \id)$, then $\Gamma_n\in C_n$ for all $n\in \N$. By compactness, a subsequence will converge to 
\begin{align*}
	\Gamma\in \bigcap_{n\in \infty} C_n =\mathcal{G}(\Gamma_0)
\end{align*}
Therefore, $G(\Gamma_0)\leq \lim_{\varepsilon\to 0} G(\Gamma_0+\varepsilon \id)$, but since $\Gamma_0\leq \Gamma_0+\varepsilon \id$ for all $\varepsilon>0$ we also have $G(\Gamma)\geq \lim_{\varepsilon\to 0} G(\Gamma_0+\varepsilon \id)$.
\end{proof}

\subsection{Additivity properties}
Finally, let us consider additivity properties of $G$. To do this, we switch our basis again and use $\gamma$ and $\sigma$.
\begin{prop} \label{prop:superadd}
For any covariance matrices $\gamma_{A}\in \R^{2n_1\times 2n_1}$ and $\gamma_B\in \R^{2n_2\times 2n_2}$, we have
\begin{align*}
	\frac{1}{2}(G(\gamma_A)+G(\gamma_B))\leq G(\gamma_A\oplus \gamma_B)\leq G(\gamma_A)+G(\gamma_B).
\end{align*}
In particular, $G$ is subadditive. 
\end{prop}
\begin{proof}
For subadditivity, let $S^TS\leq \gamma_A$ and $S^{\prime T}S^{\prime}\leq \gamma_B$ obtain the minimum in $G(\gamma_A)$ and $G(\gamma_B)$ respectively. Then $S\oplus S^{\prime}$ is symplectic and $(S\oplus S^{\prime})^T(S\oplus S^{\prime})\leq \gamma_A\oplus \gamma_B$ hence, $G(\gamma_A\oplus \gamma_B)\leq G(A)+G(B)$. 

To prove the lower bound, we need the following equation that we will only prove later on (see equation (\ref{eqn:post})):
\begin{align}
		a\geq 1: \qquad G(\gamma_A)\leq G(\gamma_A\oplus a\id_{n_2}).
\end{align} 
Assuming this inequality, let $a\geq 1$ be such that $a\id_{n_2}\geq \gamma_B$, then  
\begin{align*}
	G(\gamma_A\oplus a\id_{n_2})\leq G(\gamma_A\oplus \gamma_B)
\end{align*}
hence $G(\gamma_A)\leq G(\gamma_A\oplus \gamma_B)$ and since we can do the same reasoning for $\gamma_B$, we have $G(\gamma_A)+G(\gamma_B)\leq 2 G(\gamma_{A}\oplus \gamma_{B})$.
\end{proof}
We don't know whether $G$ is also superadditive, which would make it additive. At present, we can only prove:
\begin{cor}
Let $\gamma_{A}\in \R^{2n_1\times 2n_1}$ and $\gamma_B\in Sp(2n_2)$, be two covariance matrices (i.e. $\gamma_B$ is a covariance matrix of a pure state). Then $G$ is additive.
\end{cor}
\begin{proof}
Subadditivity has already been proven in the lemma. For superadditivity, we use the second reformulation of the squeezing measure in equation (\ref{eqn:reform2}): Note that there is only one matrix $\gamma_B\geq \gamma\geq i\sigma$, namely $\gamma_B$ itself. Now write
\begin{align*}
	\gamma_A\oplus \gamma_B\geq \begin{pmatrix}{} \tilde{A} & C \\ C^T & \tilde{B} \end{pmatrix} \geq i\sigma
\end{align*}
for $\tilde{A}\in\R^{2n_1\times 2n_1}$ and $\tilde{B}\in \R^{2n_2\times 2n_2}$. Then in particular $\gamma_B-\tilde{B}\geq 0$, but also $\tilde{B}\geq i\sigma$, hence $\gamma_B\geq \tilde{B}\geq i\sigma$ and therefore $\tilde{B}=\gamma_B$. But then
\begin{align*}
	\gamma_A\oplus \gamma_B- \begin{pmatrix}{} \tilde{A} & C \\ C^T & \tilde{B} \end{pmatrix} =\begin{pmatrix}{} \gamma_A-\tilde{A} & C \\ C^T & 0 \end{pmatrix}
\end{align*} 
hence also $C=0$ and the matrix that takes the minimum in $G(\gamma_A\oplus \gamma_B)$ must be block-diagonal. Then $\gamma_A\oplus \gamma_B\geq \tilde{A}\oplus \gamma_B\geq 0$ and $\tilde{A}$ is in the feasible set of $G(\gamma_A)$.
\end{proof}
\begin{cor} \label{cor:superadd}
For any covariance matrices $\gamma_{A}\in \R^{2n_1\times 2n_1}$ and $\gamma_B\in \R^{2n_2\times 2n_2}$,
\begin{align*}
	G(\gamma_A)+G(\gamma_B)\leq 2 G\left(\begin{pmatrix}{} \gamma_{A} & C \\ C^T & \gamma_{B}\end{pmatrix}\right).
\end{align*}
If $G$ is superadditive, then this inequality holds without the factor of two.
\end{cor}
\begin{proof}
\begin{align*}
	G(\gamma_A)+G(\gamma_B)&\leq 2 G\left(\begin{pmatrix}{} \gamma_{A} & 0 \\ 0 & \gamma_{B}\end{pmatrix}\right) = 2 G\left(\frac{1}{2}\begin{pmatrix}{} \gamma_{A} & C \\ C^T & \gamma_{B}\end{pmatrix}+\frac{1}{2}\begin{pmatrix}{} \gamma_{A} & -C \\ -C^T & \gamma_{B}\end{pmatrix}\right) \\
	&\stackrel{\mathclap{(*)}}{\leq} G\left(\begin{pmatrix}{} \gamma_{A} & C \\ C^T & \gamma_{B}\end{pmatrix}\right)+G\left(\begin{pmatrix}{} \gamma_{A} & -C \\ -C^T & \gamma_{B}\end{pmatrix}\right)~
	\stackrel{\mathclap{(**)}}{=}~2 G\left(\begin{pmatrix}{} \gamma_{A} & C \\ C^T & \gamma_{B}\end{pmatrix}\right).
\end{align*}
Here we used Proposition \ref{prop:superadd} and then convexity of $G$ in $(*)$. Finally, in $(**)$ we used that for every
\begin{align}
	\begin{pmatrix}{} \gamma_{A} & C \\ C^T & \gamma_{B}\end{pmatrix} \geq \begin{pmatrix}{} S_A & \tilde{C} \\ \tilde{C}^T & S_B \end{pmatrix}\begin{pmatrix}{} S_A & \tilde{C} \\ \tilde{C}^T & S_B \end{pmatrix}^T \in Sp(2(n_1+n_2)) \label{eqn:subaddproof1}
\end{align}
we also have:
\begin{align}
	\begin{pmatrix}{} \gamma_{A} & -C \\ -C^T & \gamma_{B}\end{pmatrix} \geq \begin{pmatrix}{} S_A & -\tilde{C} \\ -\tilde{C}^T & S_B \end{pmatrix}\begin{pmatrix}{} S_A & -\tilde{C} \\ -\tilde{C}^T & S_B \end{pmatrix}^T \in Sp(2(n_1+n_2)) \label{eqn:subaddproof2}
\end{align}
and vice versa. Since the two matrices on the right hand side of equations (\ref{eqn:subaddproof1}) and (\ref{eqn:subaddproof2}) are related by an orthogonal transformation, the two matrices have equal spectrum; hence the two squeezing measures of the matrices on the left hand side need to be equal.
\end{proof}

\subsection{Bounds} \label{sec:bounds}
Before we try to compute the squeezing measure, let us give a few simple bounds on $G$. 
\begin{prop}[spectral bounds] \label{prop:specbounds}
Let $\Gamma\geq iJ$ be a valid covariance matrix and $\lambda^{\downarrow}(\Gamma)$ be the vector of eigenvalues in decreasing order. Then:
\begin{align}
	-\frac{1}{2}\sum_{\lambda^{\downarrow}_i(\Gamma)<1}\log(\lambda^{\downarrow}_i(\Gamma))\leq G(\Gamma)\leq \frac{1}{2} \sum_{i=1}^n\log\lambda^{\downarrow}_i(\Gamma)=F(\Gamma^{1/2}) \label{eqn:specbounds}
\end{align}
\end{prop}
\begin{proof}
According to the Euler decomposition, a symplectic positive definite matrix has positive eigenvalues that come in pairs $s,s^{-1}$ and we can find $O\in SO(2n)$ such that for any $S^TS\leq \Gamma$
\begin{align*}
	O^T\Gamma O\geq \diag(s_1,\ldots,s_n,s_1^{-1},\ldots,s_n^{-1})
\end{align*}
But then, $\lambda_k^{\downarrow}(\Gamma)\geq \lambda_k^{\downarrow}(\diag(s_1,\ldots,s_n,s_1^{-1},\ldots,s_n^{-1}))$ via the Weyl inequalities $\lambda_i^{\downarrow}(A)\geq \lambda_i^{\downarrow}(B)$ for all $i$ and $A-B\geq 0$ (cf. \cite{bha96}, Theorem III.2.3). This implies:
\begin{align*}
	G(\Gamma)\leq \sum_{i=1}^n \log (\max\{s_i,s_i^{-1}\})\leq \sum_{i=1}^n \log \lambda_i^{\downarrow}(\Gamma)^{1/2}
\end{align*}
For the lower bound, given an optimal matrix $S$ with eigenvalues $s_i$, we have
\begin{align*}
	G(\Gamma)=\sum_i \max\{s_i,s_i^{-1}\}
\end{align*}
If $S^TS=O^T\diag(s_1^2,\ldots, s_n^2,s_1^{-2},\ldots,s_n^{-2})O$ with $O\in SO(2n)$ is the diagonalisation of $S^TS$, we can write:
\begin{align*}
	O^{-T}\Gamma^{-1}O^{-1}\leq \diag(s_1^2,\ldots, s_n^2,s_1^{-2},\ldots,s_n^{-2})
\end{align*}
and again by Weyl's inequalites, we can find for all $k\leq n$:
\begin{align}
	-\frac{1}{2}\sum_{i=2n-k+1}^{2n} \log(\lambda^{\downarrow}_i(\Gamma))\leq \frac{1}{2}\sum_{i=1}^k\log \lambda_i^{\downarrow}(\diag(s_1^2,\ldots, s_n^2,s_1^{-2},\ldots,s_n^{-2}))\leq G(\Gamma) \label{eqn:boundsproof}
\end{align}
Now, $-\frac{1}{2}\sum_{i=2n-k+1}^{2n} \lambda_i^{\downarrow}(\Gamma)$ can be upper bounded by restricting to eigenvalues $\lambda_i^{\downarrow}(\Gamma)<1$. This implies
\begin{align*}
	-\frac{1}{2}\sum_{\lambda_i^{\downarrow}(\Gamma)<1}\log (\lambda_i^{\downarrow}(\Gamma))\leq G(\Gamma)
\end{align*}
using that the number of eigenvalues $\lambda_i(\Gamma)<1$ can at most be $n$ (hence $k\leq n$ in the inequality of equation (\ref{eqn:boundsproof})), since $\Gamma\geq S^TS$ and $S^TS$ has at least $n$ eigenvalues bigger than one. 
\end{proof}
Numerics suggest that the lower bound is often very good for low dimensions. In fact, it can sometimes be achieved:
\begin{prop} \label{prop:achievelower}
Let $\Gamma\geq iJ$ be a covariance matrix, then $G$ achieves the lower bound in equation (\ref{eqn:specbounds}) if there exists an orthonormal eigenvector basis $\{v_i\}_{i=1}^{2n}$ of $\Gamma$ with $v_i^TJv_j=\delta_{i,n+j}$. Conversely, if $G$ achieves the lower bound, then $v_i^TJv_j=0$ for all normalised eigenvectors $v_i,v_j$ of $\Gamma$ with $\lambda_i,\lambda_j<1$. 
\end{prop}
\begin{proof}
Suppose that the lower bound in equation (\ref{eqn:specbounds}) is achieved. Via Weyl's inequalities (cf. \cite{bha96} Theorem III.2.3), for all $S^TS\leq \Gamma$ in the definition of $G$ we have $\lambda_i^{\downarrow}(S^TS)\leq \lambda_i^{\downarrow}(\Gamma)$. For the particular $S$ achieving $G$, this implies that for all $\lambda_i^{\downarrow}(\Gamma)<1$ we have $\lambda_i^{\downarrow}(S^TS)= \lambda_i^{\downarrow}(\Gamma)$, since $G$ is equal to the lower bound in (\ref{eqn:specbounds}). In particular, the smallest eigenvalue of $S^TS$ must be the same as the smallest eigenvalue of $\Gamma$. But then, it is easy to see that every eigenvector $v$ of $S^TS$ corresponding to the smallest eigenvalue must be an eigenvector of $\Gamma$ to the smallest eigenvalue since otherwise $\Gamma\not \geq S^TS$. Iteratively, we can see that if $G$ achieves the lower bound, every eigenvector of $\Gamma$ with $\lambda_i(\Gamma)<1$ must be an eigenvector of $S^TS$ with the same eigenvalue.

Since the matrix diagonalising $S^TS$ also diagonalises $\mathcal{C}^{-1}(S^TS)$, the eigenvectors of the two matrices are the same. Now, since $\mathcal{C}^{-1}(S^TS)\in\mathcal{H}$ by reformulation (\ref{eqn:reform3}), for any eigenvector $v_i$ of any eigenvalue $\mathcal{C}^{-1}(\lambda_i)<0$, $Jv_i$ is also an eigenvector of $\mathcal{C}^{-1}(S^TS)$ to the eigenvalue $-\mathcal{C}^{-1}(\lambda_i)$. Since the eigenvectors of different eigenspaces are orthogonal, this implies $v_i^TJv_j=0$ for all $i,j$. By definition, this means that $\{v_i,Jv_j\}$ forms a symplectic basis. Above, we already saw that the eigenvectors of $\Gamma$ for $\lambda_i(\Gamma)<1$ are also eigenvalues of $S^TS$, hence $v_i^TJv_j=0$ for all $i$ such that $\lambda_i(\Gamma)<1$.

Conversely, suppose we have an orthonormal basis $\{v_i\}_{i=1}^{2n}$ such that $v_i^TJv_j=\delta_{i,j+n}$ (modulo $2n$ if necessary) for all eigenvectors of $\Gamma$, i.e. $\Gamma$ is diagonalisable by a symplectic orthonormal matrix $\tilde{O}\in U(n)$. Then
\begin{align*}
	\tilde{O}\Gamma\tilde{O}^T=\diag(\lambda_1,\ldots,\lambda_{2n})
\end{align*}
Since $\Gamma\geq iJ$ we have $\lambda_i\lambda_{2i}\geq 1$. Assume that $\lambda_i\geq \lambda_{n+i}$ for all $i=1,\ldots,n$ and the $\lambda_{n+i}$ are ordered in decreasing order. Then $\lambda_{n+r}<1\leq \lambda_{n+r-1}$ for some $r\leq n$ and
\begin{align*}
	S^TS=\tilde{O}^T\diag(1,\ldots,1,\lambda_r^{-1},\ldots, \lambda_n^{-1},1,\ldots,1,\lambda_{n+r},\ldots,\lambda_{2n})\tilde{O}
\end{align*}
fulfils $S^TS\leq \Gamma$ and obviously achieves the lower bound in equation (\ref{eqn:specbounds}).
\end{proof}

In contrast to this, the upper bound can be arbitrarily bad. For instance, consider the thermal state $\Gamma=(2N+1)\cdot\id$ for increasing $N$. It can easily be seen that $G(\Gamma)=0$, since $\Gamma\geq \id\in \Pi(n)$ and $F(\id)=0$, hence $G(\Gamma)\leq 0$. However, the upper bound in equation (\ref{eqn:specbounds}) is $n/2\log(2N+1)\to \infty$ for $N\to \infty$, therefore arbitrarily bad.

We can achieve better upper bounds by using Williamson's normal form:
\begin{prop}(Williamson bounds) \label{prop:boundswill}
Let $\Gamma\in\mathbb{R}^{2n\times 2n}$ be such that $\Gamma\geq iJ$ and consider its Williamson normal form $\Gamma=S^TDS$. Then:
\begin{align}
	F(S)-\log(\sqrt{\det(\Gamma)})\leq G(\Gamma)\leq F(S) \label{eqn:boundswill}
\end{align}
\end{prop}
\begin{proof}
Since $D\geq \id$ via $\Gamma\geq iJ$, the upper bound follows directly from the definition. Also, $F(S)\leq F(\Gamma^{1/2})$, which makes this bound trivially better than the spectral upper bound in equation (\ref{eqn:specbounds}). 

The lower bound follows from:
\begin{align*}
	G(\Gamma)&\stackrel{\mathclap{(\ref{eqn:boundsproof})}}{\geq} \frac{1}{2}\log\left(\prod_{i=n+1}^{2n} \lambda_i^{\downarrow}(\Gamma)^{-1}\right) = \frac{1}{2}\log \frac{\prod_{i=1}^n\lambda_i^{\downarrow}(\Gamma)}{\prod_{i=1}^{2n}\lambda_i^{\downarrow}(\Gamma)} \\
	&=F(\Gamma^{1/2})-\log(\det(\Gamma)^{1/2})\geq F((S^TS)^{1/2})-\log(\sqrt{\det(\Gamma)}) \\
	&= F(S)-\log(\sqrt{\det(\Gamma)})
\end{align*}
using Weyl's inequalities once again, implying that since $S^TS\leq \Gamma$, we also have $F(S)^2=F(S^TS)\leq F(\Gamma)$.
\end{proof}
The upper bound here can also be arbitrarily bad. One just has to consider $\Gamma:=S^T(N\cdot \id)S$ with $S^2=\diag(N-1,\ldots,N-1,(N-1)^{-1},\ldots,(N-1)^{-1})\in Sp(2n)$. Then $\Gamma\geq \id$, i.e. $G(\Gamma)=0$, but $F(S)\to \infty$ for $N\to\infty$.
\begin{prop}
Let $\Gamma\geq iJ$ be a covariance matrix. Then 
\begin{align}
	G(\Gamma)\geq \frac{1}{4}\inf\{\|\gamma_0\|_1 | \log \Gamma \geq \gamma_0, \gamma_0\in\pi(n)\} \label{eqn:sdpbound}
\end{align}
where $\pi(n)$ was defined in Proposition \ref{prop:liealg} as the Lie algebra of the positive semidefinite symplectic matrices.
This infimum can be computed efficiently as a semidefinite program.
\end{prop}
\begin{proof}
Recall that the logarithm is operator monotone on positive definite matrices. Using this, we have:
\begin{align*}
	G(\Gamma)&= \log \inf\left\{ \prod_{i=1}^n \lambda_i^{\downarrow}(S^TS)^{1/2} \middle| \Gamma\geq S^TS\right\} \\
	&\geq \inf\left\{ \sum_{i=1}^n \log \lambda_i^{\downarrow}(\exp(\gamma_0))^{1/2}\middle|\log \Gamma\geq \gamma_0, \gamma_0\in\pi(n)\right\} \\
	&=\inf\left\{ \frac{1}{2} \sum_{i=1}^n \lambda_i^{\downarrow}(\gamma_0)\middle| \log\Gamma\geq \gamma_0, \gamma_0\in\pi(n)\right\} \\
	&=\inf\left\{ \frac{1}{4} \sum_{i=1}^{2n} s_i^{\downarrow}(\gamma_0)\middle| \log\Gamma\geq \gamma_0, \gamma_0\in\pi(n)\right\}
\end{align*}
The last step is valid, because the eigenvalues of matrices $\gamma_0\in\pi(n)$ come in pairs $\pm \lambda_i$. Since the sum of all the singular values is just the trace-norm, we are done.

It remains to see that this can be computed by a semidefinite program. First note that since the matrices $H\in \pi(n)$ are those symmetric matrices with $HJ+JH=0$, the constraints are already linear semidefinite matrix inequalities. It only remains to see that computing the trace norm is a semidefinite program. This is fairly standard \cite{rec10,van96}:
\begin{align*}
	\|\gamma_0\|_1=\min\left\{\frac{1}{2} \tr(A+B)\middle| \begin{pmatrix}{} A & \gamma_0 \\ \gamma_0 & B\end{pmatrix}\geq 0 \right\}
\end{align*}
which is clearly a semidefinite program.
\end{proof}
Numerics for small dimensions suggest that this bound is mostly smaller than the spectral lower bounds.

\section{An operational definition of the squeezing measure} \label{sec:operational}
Let us now prove that $G$ is a squeezing measure answering the question: Given a state, what is the minimal amount of single-mode squeezers needed to prepare it? In other words, it quantifies the amount of squeezing needed for the preparation of a state. Therefore, we need to specify the preparation procedure:

\subsection{Operations for state preparation and an operational measure for squeezing} \label{sub:operations}
In order to prepare a state, we need to start with a state that we can freely obtain. Obviously, such a state should not be squeezed. As many experiments are calibrated against the vacuum, starting with the vacuum ($\gamma=\id$) seems natural. Alternatively, if one is more interested in thermodynamics, the free states would be thermal states for some bath ($\gamma=(1/2+N)\id$ with photon number $N$, see e.g. \cite{oli12}). In any case, such states are not squeezed and fulfil $\gamma\geq \id$, hence we posit:
\begin{itemize}
	\item[(O0)] We can always draw $N$-mode states with $\gamma\in \mathbb{R}^{2N\times 2N}$ for any dimension $N$ from the vacuum $\gamma=\id$ or a bath fulfiling $\gamma\geq \id$.
\end{itemize}
Of course, we should be able to draw arbitrary ancillary modes of this system, too. Given a state $\rho$ with covariance matrix $\gamma$, this means that we can add modes $\gamma_{\mathrm{anc}}\geq \id$ and consider  $\gamma\oplus\gamma_{\mathrm{anc}}$ which should also be free:
\begin{itemize}
	\item[(O1)] We can always add ancillary modes from the vacuum $\gamma_{\mathrm{anc}}= \id$ or a bath $\gamma\geq \id$ and consider $\gamma\oplus\gamma_{\mathrm{anc}}$. 
\end{itemize}
A natural time evolution is not free of noise. Sometimes, it may even be beneficial to use naturally occurring noise as resource. As a noise model, let us only consider Gaussian noise, which acts on the covariance matrix $\gamma$ of our quantum state by addition of noise-covariance $\gamma_{\mathrm{noise}}\geq 0$ such that $\gamma\mapsto \gamma+\gamma_{\mathrm{noise}}$ \cite{lin00b}.
\begin{itemize}
	\item[(O2)] We can freely add noise with $\gamma_{\mathrm{noise}}\geq 0$ to our state, which is simply added to the covariance matrix of a state.
\end{itemize}
As with other squeezing measures, passive transformations should not change the squeezing measure, while single-mode-squeezers are not free. The effect of symplectic transformations on the covariance matrix has already been observed in equation (\ref{eqn:symptraf}). 
\begin{itemize}
	\item[(O3)] We can perform any beam splitter or phase shifter and in general any operation $S\in K(n)$, which translates to a map $\gamma\mapsto S^T\gamma S$ on covariance matrices of states.
	\item[(O4)] We can perform any single-mode squeezer $S=\diag(1,\ldots, 1,s,s^{-1},1\ldots,1)$ for some $s\in \R_+$.
\end{itemize}
Since we have the Weyl-system at our disposal, we can also consider its action on a quantum state. A Weyl-operator, as we have seen in equation (\ref{eqn:weylsystem}), is a translation in phase space. Direct computation shows that it does not affect the covariance matrix, but we want to include it in the set of allowed transformations nevertheless:
\begin{itemize}
	\item[(O5)] We can perform any Weyl-translation leaving the covariance matrix invariant.
\end{itemize}
Another operation which we want to consider is taking a convex combination of states. In an experiment, this can be done by creating ensembles of the states of the convex combination and creating another ensemble where the ratio of the different states is that of the convex combination. On the level of covariance matrices, we have the following lemma:
\begin{lem}
Let $\rho$ and $\rho^{\prime}$ be two states with displacement $d^{\rho}$ and $d^{\rho^{\prime}}$ and (centred) covariance matrices $\gamma^{\rho}$ and $\gamma^{\rho^{\prime}}$. For $\lambda\in (0,1)$, the covariance matrix of $\tilde{\rho}:=\lambda \rho+(1-\lambda)\rho^{\prime}$ is given by:
\begin{align*}
	\gamma^{\tilde{\rho}}=\lambda \gamma^{\rho}+(1-\lambda)\gamma^{\rho^{\prime}}+2\lambda(1-\lambda)(d^{\rho}-d^{\rho^{\prime}})(d^{\rho}-d^{\rho^{\prime}})^T
\end{align*}
\end{lem}
A proof of this statement can be found in \cite{wer01} (in the proof of proposition 1). Note that for \emph{centralized} states with $d^{\rho}=0$ and $d^{\rho^{\prime}}=0$, a convex combination of states translates to a convex combination of covariance matrices. Since in particular, $2\lambda(1-\lambda)(d^{\rho}-d^{\rho^{\prime}})(d^{\rho}-d^{\rho^{\prime}})^T\geq 0$, any convex combination of $\rho$ and $\rho^{\prime}$ is on the level of covariance matrices equivalent to 
\begin{itemize}
	\item centring the states (no change in the covariance matrices),
	\item taking a convex combination of the states (resulting in a convex combination of covariance matrices),
	\item performing a Weyl translation to undo the centralization in the first step (no change in the covariance matrix).
	\item Adding noise $2\lambda(1-\lambda)(d^{\rho}-d^{\rho^{\prime}})(d^{\rho}-d^{\rho^{\prime}})^T\geq 0$.
\end{itemize}
This implies that the effect of any convex combination of states on the covariance matrix can equivalently be obtained from operations (O2) and (O5) and a convex combination of centred states:
\begin{itemize}
	\item[(O6)] Given two (WLOG centred) states with covariance matrices $\gamma_1$ and $\gamma_2$, we can always take their convex combination $p\gamma_1+(1-p)\gamma_2$ for any $p\in [0,1]$.
\end{itemize}
Finally, we should also allow discarding certain modes of the system (which corresponds to a partial trace on the level of the density matrix) and measuring the system. The most common type of measurement in continuous variable quantum information is \emph{homodyne detection}. Homodyne detection is the measurement of $Q$ or $P$ in one of the modes, which corresponds to the measurement of an infinitely squeezed pure state in Lemma \ref{lem:measure}. A broader class of measurements known as \emph{heterodyne detection} measures arbitrary coherent states \cite{wee12}. Let us focus our attention on the even broader class of projections onto Gaussian pure states.
\begin{lem} \label{lem:measure}
Let $\rho$ be an $(n+1)$-mode quantum state with covariance matrix $\gamma$ and $|\gamma_G,d\rangle\langle \gamma_G,d|$ be a pure single-mode Gaussian state with covariance matrix $\gamma_G\in \R^{2\times 2}$ and displacement $d$. Let 
\begin{align*}
	\gamma=\begin{pmatrix}{} A & C \\ C^T & B \end{pmatrix}, \qquad B\in \R^{2\times 2}
\end{align*}
then the selective measurement of $|\gamma_G,d\rangle$ in the last mode results in a change of the covariance matrix of $\rho$ according to:
\begin{align}
	\gamma^{\prime}=A-C(B-\gamma_G)^{MP}C^{T} \label{eqn:gaussianmeasure}
\end{align}
where $^{MP}$ denotes the Moore-Penrose pseudoinverse. \emph{Homodyne detection} corresponds to the case where $\gamma_G$ is an infinitely squeezed state.
\end{lem}
This can most easily be seen on the level of Wigner functions, as demonstrated in \cite{eis02b,gie02}. The generalisation to multiple modes is straightforward.

Since the covariance matrix of a Gaussian pure state is a symplectic matrix (cf. prop \ref{prop:puregaussian}), using the Euler decomposition we can implement a selective Gaussian measurement by
\begin{enumerate}
	\item a passive symplectic transformation $S\in K(n+1)$,
	\item a measurement in the Gaussian state $\diag(d,1/d)$ for some $d\in \R_+$ according to Lemma \ref{lem:measure}.
\end{enumerate}
A non-selective measurement (forgetting the information obtained from measurement) would then be a convex combination of such projected states. A measurement of a multi-mode state can be seen as successive measurements of single-mode states since the Gaussian states we measure are diagonal. 

For homodyne detection, since an infinitely squeezed single-mode state is given by the covariance matrix $\lim_{d\to\infty} \diag(1/d,d)$, we have 
\begin{align}
	\gamma^{\prime}=\lim\limits_{d\to\infty}\left( A-C(B-\diag(1/d,d))^{-1}C^T\right)=A-C(\pi B\pi)^{MP}C^{T} \label{eqn:homodyne}
\end{align}
where $\pi=\diag(1,0)$ is a projection and $^{MP}$ denotes the Moore-Penrose-pseudoinverse. 

It has been shown (see \cite{wee12} E.2 and E.3 as well as \cite{eis02b,gie02}) that any (partial or total) Gaussian measurement is a combination of passive transformations, discarding subsystems, projection onto Gaussian states and homodyne detection.

Therefore, we should also allow to discard part of the system, i.e. taking the partial trace. However, this can be expressed as a combination of operations (O1)-(O6) and homodyne detection:
\begin{lem} \label{lem:partialtrace} 
Given a covariance matrix $\gamma=\begin{pmatrix}{} A & C \\ C^T & B\end{pmatrix}$ a partial trace on the second system translates to a map $\gamma\mapsto A$. The partial trace can then be implemented by measurements and adding noise.
\end{lem}
\begin{proof}
When measuring the modes $B$, we note that since $C(\pi B\pi)^{MP}C^{T}\geq 0$ in equation (\ref{eqn:homodyne}), a partial trace is equivalent to first performing a homodyne detection on the $B$-modes of the system and then adding noise. 
\end{proof} 

Given the discussion above, Lemma \ref{lem:measure} and Lemma \ref{lem:partialtrace} put together imply: On the level of covariance matrices, in order to allow for general Gaussian measurements, it suffices to consider Gaussian measurements of the state $|\gamma_d,0\rangle\langle \gamma_d,0|$ with covariance matrix $\gamma_d=\diag(1/d,d)$ for $d\in \R_+\cup\{+\infty\}$. All Gaussian measurements are then just combinations of these special measurements and operations (O1)-(O6):
\begin{itemize}
	\item[(O7)] At any point, we can perform a selective measurement of the system corresponding to a projection into a finitely or infinitely squeezed state. Given a state with covariance matrix $\gamma=\begin{pmatrix}{} A & B \\ B^T & C\end{pmatrix}$, this translates to a map as in equation (\ref{eqn:gaussianmeasure}) for finitely and (\ref{eqn:homodyne}) for infinitely squeezed states.
\end{itemize}

Having defined all allowed operations for preparing our quantum state, we will now define an operational measure for squeezing:
\begin{dfn}
Let $\rho$ be a quantum state with covariance matrix $\gamma$. Consider arbitrary sequences 
\begin{align*}
	\vec{\gamma}_N:=\gamma_0\to\gamma_1\to\cdots \to \gamma_N
\end{align*}
where $\gamma_0$ fulfils (O0) and every arrow corresponds to an arbitrary operation (O1)-(O5) or (O7). We then define the set of all such sequences that end in $\gamma$ by 
\begin{align*}
	\mathfrak{O}^N(\gamma)&:=\{\gamma_N=\gamma|\vec{\gamma}_N\} \\
	\mathfrak{O}(\gamma)&:=\bigcup_{N\in\N} \mathfrak{O}^N(\gamma)
\end{align*}
Furthermore, for any $\vec{\gamma}_N$, let $\vec{s}=\{s_i\}_{i=1}^M$ be the sequence of the largest singular values of any single-mode squeezer (O4) implemented along the sequence (in particular, $M\leq N$). Then we can define the measure of squeezing $G^{op}$ via
\begin{align}
	G^{op}(\rho)\equiv G^{op}(\gamma):=\inf\left\{\sum_i \log s_i\middle|s_i\in\vec{s}, \vec{\gamma}\in \mathfrak{O}(\gamma)\right\}
\end{align}
\end{dfn}
The idea of the operational measure is the following: We are allowed to prepare the state using a sequence of the allowed operations (O1)-(O7) in any order we like such that the final state is the state we would like to prepare. Then the minimal squeezing cost is the minimal amount of single-mode-squeezers that we need for any such sequence. This amount of single-mode-squeezers is measured by the logarithm of the largest singular value of the symplectic matrix needed to implement it. In order to make it more readable (it would force us to consider trees instead of sequences and introduce even more notation), we excluded convex combinations (O6) in the definition of the measure here, although this does not change anything:
\begin{thm} \label{thm:operational}
Let $\rho$ be a quantum state, then $G^{op}(\rho)=G(\rho)$. Furthermore, given $\gamma,\gamma^{\prime}$, if we allow convex combinations $\lambda\gamma+(1-\lambda)\gamma^{\prime}$ as in (O6), and assume that the costs are added according to $\lambda G(\gamma)+(1-\lambda)G(\gamma^{\prime})$, then the value of $G^{op}$ does not change.
\end{thm}
Since we consider many different operations, the proof is rather lengthy, where the main difficulties will be in showing that measurements do not squeeze. In order to increase readability, the proof will be split into several lemmata.

\subsection{Proof of the main theorem} \label{sub:proof}
\begin{lem} \label{lem:rearrange}
Let $\gamma\in \R^{2n\times 2n}$ be a covariance matrix, $\gamma_0\geq \id$, let $N\in \N$ and
\begin{align} 
	\gamma_0 \to \gamma_1\to \cdots \to \gamma_N=\gamma \label{eqn:chain}
\end{align}
be any sequence of actions (O1)-(O5) or (O7). If we denote the cost (sum of the logarithm of the largest singular values of any symplectic matrix involved) of this sequence by $c$, then one can replace this sequence by:
\begin{align}
\begin{split}
	\gamma_0 &\stackrel{(O1)}{\to} \gamma_0\oplus \gamma_{\mathrm{anc}}\stackrel{(O2)}{\to} \gamma_0\oplus \gamma_{\mathrm{anc}}+\gamma_{\mathrm{noise}} \stackrel{(O3), (O4)}{\to} S^T(\gamma_0\oplus \gamma_{\mathrm{anc}}+\gamma_{\mathrm{noise}})S \\ &\stackrel{(O7)}{\to} \mathcal{M}(S^T(\gamma_0\oplus \gamma_{\mathrm{anc}}+\gamma_{\mathrm{noise}})S) \label{eqn:normalform}
\end{split}
\end{align}
with $\gamma_{\mathrm{anc}}\geq \id$, $\gamma_{\mathrm{noise}}\geq 0$, $S\in Sp(2n)$ and $\mathcal{M}$ a partial Gaussian measurement of type specified in (O7). For this sequence, $c\geq F(S)$.
\end{lem}
\begin{proof}
We prove the proposition by proving that given any chain $\gamma_0 \to \gamma_1\to \cdots \to \gamma_N=\gamma$ as in (\ref{eqn:chain}), we can interchange all operations and obtain a chain as in equation (\ref{eqn:normalform}). For readability, we will not always specify the size of the matrices and we will assume that $\gamma\geq i\sigma$, $\gamma_{\mathrm{anc}}\geq \id$, $\gamma_{\mathrm{noise}}\geq 0$, and $S$ a symplectic matrix, whenever the symbols arise:

\begin{enumerate}
\item We can combine any sequence $\gamma_i\to\gamma_{i+1}\to\cdots \to \gamma_{i+m}$ for some $m\in\N$ where each of the arrows corresponds to a symplectic transformation $S_j$, $j=1,\ldots,m$ as in (O3) or (O4), into a single symplectic matrix $S\in Sp(2n)$ such that $\gamma_{i+m}=S^T\gamma_iS$. Furthermore Lemma \ref{lem:gelfand} implies $F(S)\leq \sum_i s^{\downarrow}_1(S_i)$, hence this recombination of steps only lowers the amount of squeezing.

\item Any sequence $\gamma\to S^T\gamma S\to S^T\gamma S+\gamma_{\mathrm{noise}}$ can be converted into a sequence $\gamma \to S^T(\gamma +\tilde{\gamma}_{\mathrm{noise}})S$ with the same $S$ and hence the same costs by setting $\tilde{\gamma}_{\mathrm{noise}}:=S^{-T}\gamma_{\mathrm{noise}}S^{-1}\geq 0$. 

\item Any sequence $\gamma\to S^T\gamma S\to S^T\gamma S\oplus \gamma_{\mathrm{anc}}$ can be converted into a sequence $\gamma\to \gamma\oplus \gamma_{\mathrm{anc}}\to \tilde{S}^T(\gamma\oplus \gamma_{\mathrm{anc}})\tilde{S}$ by setting $\tilde{S}=S\oplus \id$ with $\id$ of the same dimension as $\gamma_{\mathrm{anc}}$. Since we only add the identity, we have $F(\tilde{S})=\sum_i \log s^{\downarrow}_i(\tilde{S})=F(S)$ and the costs do not increase.

\item Any sequence $\gamma\to \gamma +\gamma_{\mathrm{noise}}\to (\gamma +\gamma_{\mathrm{noise}})\oplus \gamma_{\mathrm{anc}}$ can be converted into a sequence $\gamma\to \gamma \oplus\gamma_{\mathrm{anc}}\to \gamma\oplus \gamma_{\mathrm{anc}} +\tilde{\gamma}_{\mathrm{noise}}$ by setting $\tilde{\gamma}_{\mathrm{noise}}=\gamma_{\mathrm{noise}}\oplus 0\geq 0$, which is again a valid noise matrix. As no operation of type (O4) is involved, the squeezing costs do not change.
\end{enumerate}
In a next step we consider measurements. We will only consider homodyne detection, since the proof is exactly the same for arbitrary Gaussian measurements of type (O7). Given a covariance matrix $\gamma$, we assume a decomposition 
\begin{align*}
	\gamma=\begin{pmatrix} A & C \\ C^T & B\end{pmatrix}; \qquad \mathcal{M}(\gamma)=A-C(\pi B\pi)^{MP}C^T
\end{align*}
as in the definition of (O7) with $\pi=\diag(1,0)$. 

\begin{enumerate}
\setcounter{enumi}{4}
\item Any sequence $\gamma\to \mathcal{M}(\gamma)\to S^T\mathcal{M}(\gamma) S$ can be converted into a sequence $\gamma\to \tilde{S}^T\gamma \tilde{S}\to \mathcal{M}(\tilde{S}^T\gamma \tilde{S})$ by setting $\tilde{S}=S\oplus \id_2$. To see this, write $S^T\mathcal{M}(\gamma) S=S^TAS-S^TC(\pi B\pi)^{MP}C^TS$ and 
\begin{align*}
	\mathcal{M}\left(\begin{pmatrix}{} S & 0 \\ 0 & \id\end{pmatrix}^T\begin{pmatrix} A & C \\ C^T & B\end{pmatrix}\begin{pmatrix}{} S & 0 \\ 0 & \id\end{pmatrix}\right)&=\mathcal{M}\left(\begin{pmatrix} S^TAS & S^TC \\ C^TS & B\end{pmatrix}\right)\\ &=S^TAS-S^TC(\pi B\pi)^{MP}C^TS
\end{align*}
hence the final covariance matrices are the same. By the same reasoning as in 3., the costs are equivalent.

\item Any sequence $\gamma\to \mathcal{M}(\gamma)\to \mathcal{M}(\gamma)+\gamma_{\mathrm{noise}}$ can be converted into a sequence $\gamma\to \gamma+\tilde{\gamma}_{\mathrm{noise}}\to \mathcal{M}(\gamma+\tilde{\gamma}_{\mathrm{noise}})$ by setting $\tilde{\gamma}_{\mathrm{noise}}=\gamma_{\mathrm{noise}}\oplus 0$, with $0$ on the last mode being measured. Since no symplectic matrices are involved, the costs are equivalent.

\item Any sequence $\gamma\to \mathcal{M}(\gamma)\to \mathcal{M}(\gamma)\oplus \gamma_{\mathrm{anc}}$ can be changed into a sequence $\gamma\to \gamma\oplus\gamma_{\mathrm{anc}} \to \tilde{\mathcal{M}}(\gamma\oplus\gamma_{\mathrm{anc}})$, where the measurement $\tilde{\mathcal{M}}$ measures the last mode of $\gamma$, i.e.
\begin{align*}
	\tilde{M}\left( \begin{pmatrix}{} A & C & 0 \\ C^T & B & 0 \\ 0 & 0 & \gamma_{\mathrm{anc}}\end{pmatrix}\right)= (A\oplus \gamma_{\mathrm{anc}})-(C\oplus 0)(\pi B \pi)^{MP}(C\oplus 0)^T
\end{align*}
Clearly, the resulting covariance matrices of the two sequences are the same and the costs are equivalent.
\end{enumerate}

We can now easily prove the lemma. Let $\gamma_0\to\ldots \to \gamma_n$ be an arbitrary sequence with operations of type (O1)-(O5) or (O7). We can first move all measurements to the right of the sequence, i.e. we first perform all operations of type (O1)-(O5) and then all measurements. This is done using the observations above: the seventh for (O1), the sixth for (O2) and the fifth for (O3) and (O4). Since (O5) does not change the covariance matrix at all, we can completely neglect them. We have also seen that the resulting sequence has the same squeezing costs. Note also that this step is similar to the quantum circuit idea to ``perform all measurements last'' (cf. \cite{nie00}, chapter 4).

Then, we can move all symplectic operations to the right (perform all operations (O1)-(O2) ahead of all operations (O3)-(O4)) using the second observation for (O1) and the third for (O2), which does not change the squeezing costs, either.

Using the forth observation, we can switch all operations (O1) to the beginning of the chain, which also does not change the squeezing costs and using the first observation we can combine all operations (O3)-(O4) into one application of a symplectic matrix $S$. This might reduce the squeezing cost. 

All in all, we obtain a new sequence as in equation (\ref{eqn:normalform}) with at most the costs of the sequence $\gamma_1\to \cdots \to \gamma_m$ we started with.
\end{proof}

We can now slowly work towards Theorem \ref{thm:operational}:
\begin{lem} \label{lem:mainproof1}
Let $\gamma\in\R^{2n\times 2n}$ be a covariance matrix, then
\begin{align}
	G(\gamma)=\inf\{ F(S)| \gamma=S^T(\gamma_0\oplus \gamma_{\mathrm{anc}}+\gamma_{\mathrm{noise}})S,~S\in Sp(2n),~\gamma_0\oplus\gamma_{\mathrm{anc}}\geq \id,~\gamma_{\mathrm{noise}}\geq 0\}
\end{align}
\end{lem}
\begin{proof}
First note that for any $\gamma\geq i\sigma$, we can find $S\in Sp(2n)$, $\gamma_0\in \R^{2n\times 2n}$ with $\gamma_0\geq \id$ and $\gamma_{\mathrm{noise}}\in\R^{2n\times 2n}$ with $\gamma_{\mathrm{noise}}\geq 0$ such that $\gamma=S^T(\gamma_0+\gamma_{\mathrm{noise}})S$ by using Williamson's Theorem, hence the feasible set is never empty. The lemma is immediate by observing that for any $\gamma=S^T(\gamma_0\oplus \gamma_{\mathrm{anc}}+\gamma_{\mathrm{noise}})S$ since $(\gamma_0\oplus \gamma_{\mathrm{anc}}+\gamma_{\mathrm{noise}})\geq \id$ we have $\gamma\geq S^TS$ and conversely, for any $\gamma\geq S^TS$, defining $\gamma_0:=S^{-T}\gamma S^{-1}\geq \id$, we have $\gamma=S^T\gamma_0S$. 
\end{proof}

As an intermediate step towards the Theorem, define:
\begin{dfn}
For $\gamma\in\R^{2n\times 2n}$ a covariance matrix, define 
\begin{align}
\begin{split}
	\tilde{G}^{\mathrm{op}}(\gamma):=\inf&\left\{F(S)|\gamma=\mathcal{M}(S^T(\gamma_0\oplus \gamma_{\mathrm{anc}}+\gamma_{\mathrm{noise}})S),~S\in Sp(2n),\right. \\
	&\left.\qquad \quad \gamma_0\oplus\gamma_{\mathrm{anc}}\geq \id_{2n},~\gamma_{\mathrm{noise}}\geq 0,\mathcal{M}~\mathrm{measurement}\right\}
\end{split}
\end{align}
\end{dfn}

Then we have:
\begin{lem} \label{lem:mainproof2}
For $\gamma\in\R^{2n\times 2n}$ a covariance matrix, we have
\begin{align}
	\tilde{G}^{\mathrm{op}}(\gamma)=\inf\{F(\hat{\gamma}^{1/2})|\gamma=\mathcal{M}(\tilde
{\gamma}),\tilde{\gamma}\geq \hat{\gamma}\geq i\sigma,~\mathcal{M}~\mathrm{
measurement}\}
\end{align}
\end{lem}
\begin{proof}
This follows from Lemma \ref{lem:mainproof1}:
\begin{align}
	\tilde{G}^{\mathrm{op}}(\gamma)&=\inf\{F(S)|\gamma=\mathcal{M}(S^T(\gamma_0\oplus \gamma_{\mathrm{anc}}+\gamma_{\mathrm{noise}})S)\} \nonumber \\
	&=\inf\{F(S)|\gamma=\mathcal{M}(\tilde{\gamma}),\tilde{\gamma}=S^T(\gamma_0 \oplus \gamma_{\mathrm{anc}}+\gamma_{\mathrm{noise}})S\geq i\sigma \} \nonumber \\
	&\stackrel{\mathclap{\mathrm{Lemma}~\ref{lem:mainproof1}}}{=}\quad~ \inf\{G(\tilde{\gamma})|\gamma=\mathcal{M}(\tilde{\gamma}), \tilde{\gamma}\geq i \sigma\} \label{eqn:postpost} \\
	&\stackrel{\mathclap{\mathrm{Prop.}~\ref{prop:reform}}}{=} \quad~ \inf\{F(
\hat{\gamma}^{1/2})|\gamma=\mathcal{M}(\tilde{\gamma}), \tilde{\gamma}\geq\hat
{\gamma}\geq i\sigma\} \nonumber
\end{align}
by taking the infimum over all measurements last. 
\end{proof}
Note here, that equation (\ref{eqn:postpost}) together with the following Proposition \ref{prop:mainproof3} finishes the proof of Proposition \ref{prop:superadd} via:
\begin{align}
	G(\gamma)= \inf\{G(\tilde{\gamma})| \gamma=\mathcal{M}(\tilde{\gamma}), \tilde{\gamma}\geq i
\sigma\}\leq G(\gamma\oplus a\id_{n_2}) \label{eqn:post}
\end{align}
for $a\geq 1$, using that measuring the last modes we obtain $\mathcal{M}(\gamma\oplus a\id_{n_2})=\gamma$ and therefore, $\gamma\oplus a\id_{n_2}$ is in the feasible set of $\tilde{G}^{\mathrm{op}}(\gamma)=G(\gamma)$.

\begin{prop} \label{prop:mainproof3}
For $\gamma\in\R^{2n\times 2n}$ a covariance matrix we have 
\begin{align*}
	\tilde{G}^{\mathrm{op}}(\gamma)=G(\gamma)
\end{align*}
\end{prop}
This proposition shows that $G$ is operational if we exclude convex combinations (and therefore also non-selective measurements).
\begin{proof}
Using Lemma \ref{lem:mainproof2}, the proof of this proposition reduces to the question whether:
\begin{align}
	\inf\{F(\hat{\gamma}^{1/2})|\tilde{\gamma}\geq \hat{\gamma}\geq i\sigma, \mathcal{M}(\tilde{\gamma})=\gamma\}
	=\inf\{F(\overline{\gamma}^{1/2})|\gamma\geq \overline{\gamma}\geq i\sigma\} \label{eqn:operationproof1}
\end{align}
Since we do not need to use measurements, $\leq$ is obvious. The crucial part will be proving $\geq$, which is equivalent to saying that measurements cannot squeeze. Similar observations have been made in papers about squeezing (see for instance \cite{kra03b}), but this only ever referred to the smallest eigenvalue of the covariance matrix, whereas we need to have control over the product of several eigenvalues here.

Let $\tilde{\gamma}\geq \hat{\gamma}\geq i\sigma$ for some $\mathcal{M}(\tilde{\gamma})=\gamma$. Our first claim is that
\begin{align}
	\gamma\geq \mathcal{M}(\hat{\gamma})\geq i\sigma \label{eqn:operationproof2}
\end{align}
$\mathcal{M}(\hat{\gamma})\geq i\sigma$ is clear from the fact that $\hat{\gamma}$ is a covariance matrix and a measurement takes states to states. $\gamma\geq \mathcal{M}(\hat{\gamma})$ is proved using \emph{Schur complements}. Let $\mathcal{M}$ be a Gaussian measurement as in equation (\ref{eqn:gaussianmeasure}) with $\gamma_G=\diag (d,1/d)$ with $d\in \R^+$. It is well-known that 
\begin{align*}
	(\id\oplus \diag(1/d, d)&\gamma (\id\oplus \diag(1/d,d))+0\oplus \id_2)^S \\
	&=A-C \diag(1/d,d)(\diag(1/d,d) B \diag(1/d,d)+\id)^{-1}\diag(1/d,d) C^T \\
	&=A-C(B+\diag(d,1/d))^{-1}C^T=\mathcal{M}(\gamma)
\end{align*}
where $^S$ denotes the Schur complement of the block in the lower-right corner of the matrix and we decompose 
\begin{align*}
	\gamma=\begin{pmatrix} A & C \\ C^T & B\end{pmatrix}
\end{align*}
as usual. For homodyne measurements, we take the limit $d\to \infty$. Since for any $\tilde{\gamma}\geq \hat{\gamma}\geq 0$, the Schur complements of the lower right block fulfil $\tilde{\gamma}^S\geq \hat{\gamma}^S\geq 0$ (cf. \cite{bha07}, exercise 1.5.7), we have $\gamma\geq \mathcal{M}(\hat{\gamma})$ as claimed in equation (\ref{eqn:operationproof2}).

Next, we claim
\begin{align}
	F(\mathcal{M}(\hat{\gamma})^{1/2})\leq F(\hat{\gamma}^{1/2}) \label{eqn:operationproof3}
\end{align}
To prove this claim, note that via the monotonicity of the exponential function on $\R$, it suffices to prove
\begin{align*}
	\prod_{j=1}^m s^{\downarrow}_j(\mathcal{M}(\hat{\gamma}))\leq \prod_{j=1}^n s^{\downarrow}_j(\hat{\gamma})
\end{align*}
when we assume $\hat{\gamma}\in\R^{2n\times 2n}$ and $\mathcal{M}(\hat{\gamma})\in \R^{2m\times 2m}$ with $m\leq n$. Again, we write
\begin{align*}
	\hat{\gamma}=\begin{pmatrix}{} \hat{A} & \hat{C} \\ \hat{C}^T & \hat{B}\end{pmatrix}
\end{align*} 
then the state after measurement is given by $\mathcal{M}(\hat{\gamma})=\hat{A}-\hat{C}(\hat{B}+\diag(d,1/d))^{-1}\hat{C}^T$ or the limit $d\to \infty$ for homodyne measurements. In any case $\hat{C}(\hat{B}+\diag(d,1/d))^{-1}\hat{C}^T\geq 0$ and $\mathcal{M}(\hat{\gamma})\leq \hat{A}$ and therefore, by Weyl's inequalities, also
\begin{align*}
	\prod_{j=1}^m s^{\downarrow}_j(\mathcal{M}(\hat{\gamma}))\leq \prod_{j=1}^m s^{\downarrow}_j(\hat{A})
\end{align*}
Now we use Cauchy's interlacing theorem (cf. \cite{bha96}, Corollary III.1.5): As $\hat{A}$ is a submatrix of $\hat{\gamma}$, we have $\lambda_i^{\downarrow}(\hat{A})\leq \lambda_i^{\downarrow}(\hat{\gamma})$ for all $i=1,\ldots,2m$. Since at least $m$ eigenvalues of $\hat{A}$ are bigger or equal one and at least $n$ eigenvalues of $\hat{\gamma}$ are bigger or equal one, this implies 
\begin{align}
	\prod_{j=1}^m s_j^{\downarrow}(\hat{A})=\prod_{j=1}^m \lambda_j^{\downarrow}(\hat{A})\leq \prod_{j=1}^m \lambda_j^{\downarrow}(\hat{\gamma}) \leq \prod_{j=1}^n \lambda_j^{\downarrow}(\hat{\gamma})=\prod_{j=1}^n s_j^{\downarrow}(\hat{\gamma}) \label{eqn:needforg}
\end{align}
In particular, this proves equation (\ref{eqn:operationproof3}).


We can then complete the proof: Let $\tilde{\gamma}\geq \hat{\gamma}\geq i\sigma$ for some $\mathcal{M}(\tilde{\gamma})=\gamma$ in equation (\ref{eqn:operationproof1}). We have just seen that this implies $\gamma\geq \mathcal{M}(\hat{\gamma})\geq i\sigma$ via equation (\ref{eqn:operationproof2}) and furthermore that $F(\hat{\gamma}^{1/2})\geq F(\mathcal{M}(\hat{\gamma})^{1/2})$ via equation (\ref{eqn:operationproof3}). But this means that we have found $\overline{\gamma}:=\mathcal{M}(\hat{\gamma})$ such that $\gamma\geq \overline{\gamma}\geq i\sigma$. Hence $\overline{\gamma}$ is in the feasible set of the right hand side of (\ref{eqn:operationproof1}) and $F(\tilde{\gamma}^{1/2})\geq F(\overline{\gamma}^{1/2})$, which implies $\geq$ in equation (\ref{eqn:operationproof1}).
\end{proof}

Finally, we can prove Theorem \ref{thm:operational} by also covering convex combinations:
\begin{proof}
Let $\gamma\in\R^{2n\times 2n}$ be a covariance matrix. In the definition of $G^{\mathrm{op}}$, we considered all possible sequences of operations (O1)-(O7) (excluding convex combinations as in (O6)). Using Lemma \ref{lem:mainproof1}, we can replace these sequences by a very special type of sequences (first (O1), then (O2), then (O3) and (O4), then (O7), (O5) can be left out). For these sequences, we have seen that the minimum cost is given by $G(\gamma)$ in Proposition \ref{prop:mainproof3}. Hence, Proposition \ref{prop:mainproof3} actually already proves $G^{\mathrm{op}}(\gamma)=G(\gamma)$. 

However, we explicitly excluded convex combinations (O6) from the definition of $G^{\mathrm{op}}$, since allowing convex combinations forces us to consider trees instead of sequences in the definition of $G^{\mathrm{op}}$: Consider a tree of operations (O1)-(O7) which has $\gamma$ at its root and $\gamma_0=\id$ as leaves (i.e. the natural generalisation of sequences $\gamma_0\to\cdots\to\gamma_N=\gamma$ including convex combinations). Let us consider any node closest to the leaves. At such a node, we start with two covariance matrices $\gamma_1$ and $\gamma_2$ that were previously constructed without using convex combinations and with costs $G(\gamma_1)$ and $G(\gamma_2)$. The combined matrix would be $\tilde{\gamma}:=\lambda \gamma_1+(1-\lambda)\gamma_2$ for some $\lambda\in (0,1)$ and the costs would be $\lambda G(\gamma_1)+(1-\lambda)G(\gamma_2)$. 

By convexity of $G$ (see Theorem \ref{thm:convex}):
\begin{align*}
	G(\lambda \gamma_1+(1-\lambda)\gamma_2)\leq \lambda G(\gamma_1)+(1-\lambda) G(\gamma_2)
\end{align*}
which means that we can find a sequence (without any convex combinations) producing $\lambda \gamma_1+(1-\lambda)\gamma_2$ which is cheaper than first producing $\gamma_1$ and $\gamma_2$ and then taking a convex combination. Iteratively, this means we can eliminate every node from the tree and replace the tree by a sequence of operations (O1)-(O5) and (O7), which is cheaper than the tree. Therefore, we can conclude that an inclusion of convex combinations as in (O6) cannot change $G^{\mathrm{op}}(\gamma)$ for any $\gamma$. Therefore, $G^{\mathrm{op}}(\gamma)$ is already operational for all operations (O1)-(O7) and any operation that can be constructed as a mixture thereof.
\end{proof}

\subsection{The squeezing measure as a resource measure} \label{sec:resourcemeasure}
We have now seen that the measure $G$ can be interpreted as a measure of the amount of single-mode squeezing needed to create a state $\rho$. Let us now take a different perspective, which is the analogue of the entanglement of formation for squeezing: Consider covariance matrices of the form
\begin{align}
	\gamma_s:=\begin{pmatrix}{} s & 0 \\ 0 & s^{-1} \end{pmatrix}
\end{align}
These are single-mode squeezed states with squeezing parameter $s\geq 1$. We will now allow these states as \emph{resources} and ask the question: Given a (Gaussian) state $\rho$ with covariance matrix $\gamma$, what is the minimal amount of these resources needed to construct $\gamma$, if we can freely transform the state by
\begin{enumerate}
	\item passive transformations,
	\item adding ancillas in the vacuum,
	\item adding noise,
	\item displacing the state (Weyl displacements),
	\item performing partial or total Gaussian measurements,
	\item and taking convex combinations of states.
\end{enumerate}
This implies that we define the following measure:
\begin{dfn}
Let $\rho$ be an $n$-mode state with covariance matrix $\gamma\in \R^{2n\times 2n}$. Let 
\begin{align}
	G^{\mathrm{resource}}(\gamma):=\inf\left\{\sum_{i=1}^m \frac{1}{2} \log(s_m) \middle| \gamma=\mathcal{T}\left(\bigoplus_{i=1}^m\gamma_{s_i}\right)\right\}
\end{align}
where $\mathcal{T}:\R^{2m\times 2m}\to \R^{2n\times 2n}$ is a combination of the operations 1-6 above. 
\end{dfn}
In principle, we could also add ancillary resource states later on, thereby enlarging the class of allowed transformations $\mathcal{T}$ in the definition, however using the same proof as in Lemma \ref{lem:rearrange}, it is easy to see that this does not change the measure.

We introduced the factor $1/2$ in the definition of $G^{\mathrm{resource}}$ in order to have the following easy characterisation:
\begin{thm}
Let $\rho$ be an $n$-mode state with covariance matrix $\gamma\in \R^{2n\times 2n}$. Then
\begin{align}
	G^{\mathrm{resource}}(\gamma)=G(\gamma)
\end{align}
\end{thm}
\begin{proof}
$\leq$: Note that for any feasible $S\in Sp(2n)$ in $G(\gamma)$, i.e. any $S$ with $S^TS\leq \gamma$, we can find $O\in Sp(2n)\cap O(2n)$ and $D=\bigoplus_{i=1}^n \gamma_{s_i}$ with $S^TS=O^TDO$ via the Euler decomposition. Using that the Euler decomposition minimises $F$, we have $F(S)=\frac{1}{2}F(D)=\sum_{i=1}^n \frac{1}{2} \log(s_i)$. But then, since we can find $\gamma_{\mathrm{noise}}\geq 0$ such that $\gamma=O^T\bigoplus_{i=1}^n \gamma_{s_i}O+\gamma_{\mathrm{noise}}$, we have that $D$ is a feasible resource state to produce $\gamma$. This implies $G^{\mathrm{resource}}(\gamma)\leq G(\gamma)$.

$\geq$: For the other direction, the proof proceeds exactly as the proof of Theorem \ref{thm:operational}. First, we exclude convex combinations. Then, we realize that we can change the order of the different operations (even if we include adding resource states during any stage of the preparation process) according to Lemma \ref{lem:rearrange}, making sure that any preparation procedure can be implemented via:
\begin{align*}
	\gamma=\mathcal{M}\left(O\left(\bigoplus_{i=1}^m\gamma_{s_i}\oplus \id_{2m^{\prime}}+\gamma_{\mathrm{noise}}\right)O^T\right)
\end{align*}
where $O\in Sp(2m+2m^{\prime})\cap O(2m+2m^{\prime})$, $\gamma_{\mathrm{noise}}\in \R^{2m+2m^{\prime}\times 2m+2m^{\prime}}$ with $\gamma_{\mathrm{noise}}\geq 0$ and $\mathcal{M}$ a measurement. Now the only difference to proof of \ref{thm:operational} is that we had the vacuum $\id$ instead of $\bigoplus_{i=1}^m\gamma_{s_i}\oplus \id_{2m^{\prime}}$ and an arbitrary symplectic matrix $S$ instead of $O$, but the two ways of writing the maps are completely interchangeable, so that the proof proceeds as in Theorem \ref{thm:operational}.
\end{proof}

\section{Calculating the squeezing measure} \label{sec:examples}
We have seen that the measure $G$ is operational. However, to be useful, we need a way to compute it.

\subsection{Analytical solutions}
\begin{prop}
Let $n=1$, then $G(\Gamma)=-\frac{1}{2}\min_i \log(\lambda_i(\Gamma))$ for all $\Gamma\in \R^{2n\times 2n}$.
\end{prop}
\begin{proof}
Note that this is the lower bound in Proposition \ref{prop:specbounds}, hence $-\frac{1}{2}\min_i \log(\lambda_i(\Gamma))\leq G(\Gamma)$. Now consider the diagonalisation $\Gamma=O\diag(\lambda_1,\lambda_2)O^T$ with $O\in SO(2)$ and assume $\lambda_1\geq \lambda_2$. Then, $\lambda_2^{-1}\leq \lambda_1$ since otherwise, $\Gamma\not\geq iJ$. 

Consider $\diag(\lambda_1,\lambda_2)\geq O^{-T}S^TSO^{-1}$ for some $S\in Sp(2)$ with eigenvalues $s\geq 1$ and $s^{-1}$. Since $\diag(\lambda_1,\lambda_2)\geq O^{-T}S^TSO^{-1}$, this implies in particular that $s^{-1}\leq \lambda_2$ by Weyl's inequality. Since $F(S^TS)=\log s$, in order to minimise $F(S)$ over $S^TS\leq \Gamma$, we need to maximize $s^{-1}$. Setting $s^{-1}=\lambda_2$ we obtain $s=\lambda_2^{-1}\leq \lambda_1$ and $\diag(\lambda_1,\lambda_2)\geq \diag(s,s^{-1})$. Since $SO(2)=K(1)$, $S^TS:=O^T\diag(\lambda_1,\lambda_2)O\leq \Gamma$ is the minimising matrix in $G$ and $G(\Gamma)=F(S)=\frac{1}{2}\log \lambda_2^{-1}$. 
\end{proof}
\begin{prop}
Let $\rho$ be a pure, Gaussian state with covariance matrix $\Gamma\in \R^{2n\times 2n}$. Then $G(\Gamma)=F(\Gamma^{1/2})$.
\end{prop}
\begin{proof}
From Proposition \ref{prop:puregaussian}, we know that $\det(\Gamma)=1$ in particular. Therefore, the bounds in Proposition \ref{prop:boundswill} are tight and $G(\Gamma)=F(\Gamma^{1/2})$.
\end{proof}

\subsection{Numerical calculations using Matlab}
The crucial observation to numerically find the optimal squeezing measure and a symplectic matrix $S$ at the optimal point is given in Lemma \ref{lem:convf}: If we use $G$ in the form of equation (\ref{eqn:reform3}), we know that the function to be minimised is convex on $\mathcal{H}$. In general, convex optimization with convex constraints is efficiently implementable and there is a huge literature on the topic (see \cite{boy04} for an overview). 

In our case, a certain number of problems occur when performing convex optimization:
\begin{enumerate}
	\item The function $f$ in equation (\ref{eqn:convexsum}) is highly nonlinear. It is also not differentiable at eigenvalue crossings of $A+iB$ or $H\in\mathcal{H}$. In particular, it is not differentiable when one of the eigenvalues becomes zero, which is to be expected at the minimum.
	\item While the constraints $\mathcal{C}^{-1}(\gamma)\geq H$ and $\id > H > -\id$ are linear in matrices, they are nonlinear in simple parameterisations of matrices.
	\item For $\gamma$ on the boundary of the set of allowed density operators, the set of feasible solutions might not have an inner point.
\end{enumerate}
The first and second problem imply that most optimization methods (including the standard ones of the \textsc{Matlab}-optimization toolbox) are unsuitable, as they are either gradient-based or need more problem structure. It also means that there is no guarantee for good stability of the solutions. The third problem implies that interior point methods become unsuitable on the boundary. Since many interesting states studied in the literature usually do have symplectic eigenvalues equal to one, this would limit the applicability and usefulness of the program.

As a proof of principle implementation, we used the \textsc{Matlab}-based solver \textsc{SolvOpt}, which can solve nonsmooth, nonlinear optimization problems with nonlinear constraints based on a penalty method (for details see the manual \cite{kun97}). The optimization used by \textsc{SolvOpt} is sub-gradient based and requires the objective function to be differentiable almost everywhere, which is true in our case. We believe our implementation could be made more efficient and more stable, but it seems to work well in most cases for less than ten modes. More information including the details of the implementation as well as the source-code are provided in appendix \ref{app:code}. 

\subsection{Squeezing-optimal preparation for certain three-mode separable states}
Let us now work with a particular example that has been studied in the quantum information literature. In \cite{mis08}, Mi\v{s}ta Jr. and Korolkova define the following three-parameter group of three-mode states where the modes are labeled $A,B,C$:
\begin{align}
	\gamma=\gamma_{AB}\oplus \id_{C}+x(q_1q_1^T+q_2q_2^T) \label{eqn:korolkovastate}
\end{align}
with
\begin{align*}
	\gamma_{AB}&=\begin{pmatrix}{} 
		e^{2d}a & 0 & -e^{2d}c & 0 \\ 0 & e^{-2d}a & 0 & e^{-2d}c \\
		-e^{2d}c & 0 & e^{2d}a & 0 \\ 0 & e^{-2d}c & 0 & e^{-2d}a
	\end{pmatrix} \\
	q_1&=(0,\sin \phi, 0, -\sin \phi, \sqrt{2},\sqrt{2})^T \\
	q_2&=(\cos \phi, 0, \cos \phi, 0 \sqrt{2},\sqrt{2})^T
\end{align*}
where $a=\cosh(2r)$, $c=\sinh(2r)$, $\tan\phi = e^{-2r}\sinh(2d)+\sqrt{1+e^{-4r}\sinh^2(2d)}$. The remaining parameters are $d\geq r>0$ and $x\geq 0$. For 
\begin{align*}
	x=x_{\mathrm{sep}}\geq \frac{2\sinh(2r)}{e^{2d}\sin^2\phi+e^{-2d}\cos^2\phi}
\end{align*}
the state becomes fully separable \cite{mis08}. The state as such is a special case of a bigger family described in \cite{gie01b}. In \cite{mis08}, it was used to entangle two systems at distant locations using fully separable mediating ancillas (here the system labeled $C$). Therefore, Mi\v{s}ta Jr. and Korolkova considered also an LOCC procedure to prepare the state characterised by (\ref{eqn:korolkovastate}). For our purposes, this is less relevant and we allow for arbitrary preparations of the state. This was also done in \cite{mis08} by first preparing modes $A$ and $B$ each in a pure squeezed-state with position quadratures $e^{2(d-r)}$ and $e^{2(d+r)}$. A vacuum mode in $C$ was added and $x(q_1q_1^T+q_2q_2^T)$ was added as random noise. Therefore, the squeezing needed to produce this state in this protocol is given by
\begin{align}
	c=\frac{1}{2}\log(e^{2(d-r)}\cdot e^{2(d+r)})=2d \label{eqn:corsqueeze}
\end{align}
We numerically approximated the squeezing measure for $\gamma_{ABC}$, choosing $x=x_{\mathrm{sep}}$, which leaves a two-parameter family of states. We chose parameters $d$ and $r$ according to 
\begin{align}
	r=0.1+j\cdot 0.05, \qquad d=r+i\cdot 0.03 \label{eqn:examplenumerics}
\end{align}
with $i,j\in\{1,\ldots,30\}$ for a total of 900 data points. Since the algorithm is not an interior point algorithm as described above, to check the result, we reprepared the state in the following way:
\begin{enumerate}
	\item Let $S$ be the symplectic matrix at the value optimum found by \textsc{SolvOpt} for a covariance matrix $\gamma_{ABC}$. 
	\item Calculate $S^{-T}\gamma_{ABC}S^{-1}$ and calculate its lowest eigenvalue $\lambda_{2n}$. 
	\item Define $\tilde{\gamma}:=S^{-T}\gamma_{ABC}S^{-1}+(1-\min\{1,\lambda_{2n}\})\id\geq \id$. Calculate the largest singular value of $S^T\tilde{\gamma}S-\gamma$. 
\end{enumerate}
If $S$ was a feasible point, then $S^T\tilde{\gamma}S=\gamma$. Since it is obvious how to prepare $\tilde{\gamma}$ with operations specified in section \ref{sec:operational}, the largest singular value of $S^T\tilde{\gamma}S-\gamma$ is an indicator of how well we can approximate the state we want to prepare by a state with comparably low squeezing costs. One can easily see that $\gamma_{ABC}$ cannot achieve the spectral lower bound, i.e. the assumptions of Lemma \ref{prop:achievelower} are not met.

The results of the numerical computation are shown in figure \ref{fig:korolkova}. We computed the minimum both with the help of numerical and analytical subgradients (see appendix \ref{app:code}) and took the value with a better approximation error. Usually, the proposed minimum as well as the preparation error of the two algorithms were extremely close (the difference was in the sub-permille regime), while at rare occasions, one algorithm failed to obtain a minimum (luckily, we never had an instance where all algorithms failed). Possible reasons for this are discussed in appendix \ref{app:code}. The optimal values computed by the algorithm are close to the lower bound and a lot better than the upper bound and the costs obtained by equation (\ref{eqn:corsqueeze}). The preparation error is also usually very small ($\mathcal{O}(10^{-7})$), and therefore we are allowed to conclude that the result is a good approximation to the real optimum.

Let us point out that the states in equation (\ref{eqn:korolkovastate}) have one symplectic eigenvalue equal to one and the feasible set has no interior point in our parameterisation, which makes it impossible to apply interior point methods. Further comments regarding errors and stability problems with the algorithm are discussed in appendix \ref{app:code}.

\begin{figure}[htbp]
\centering
	\includegraphics[trim=1cm 9cm 6cm 10cm,width=14cm]{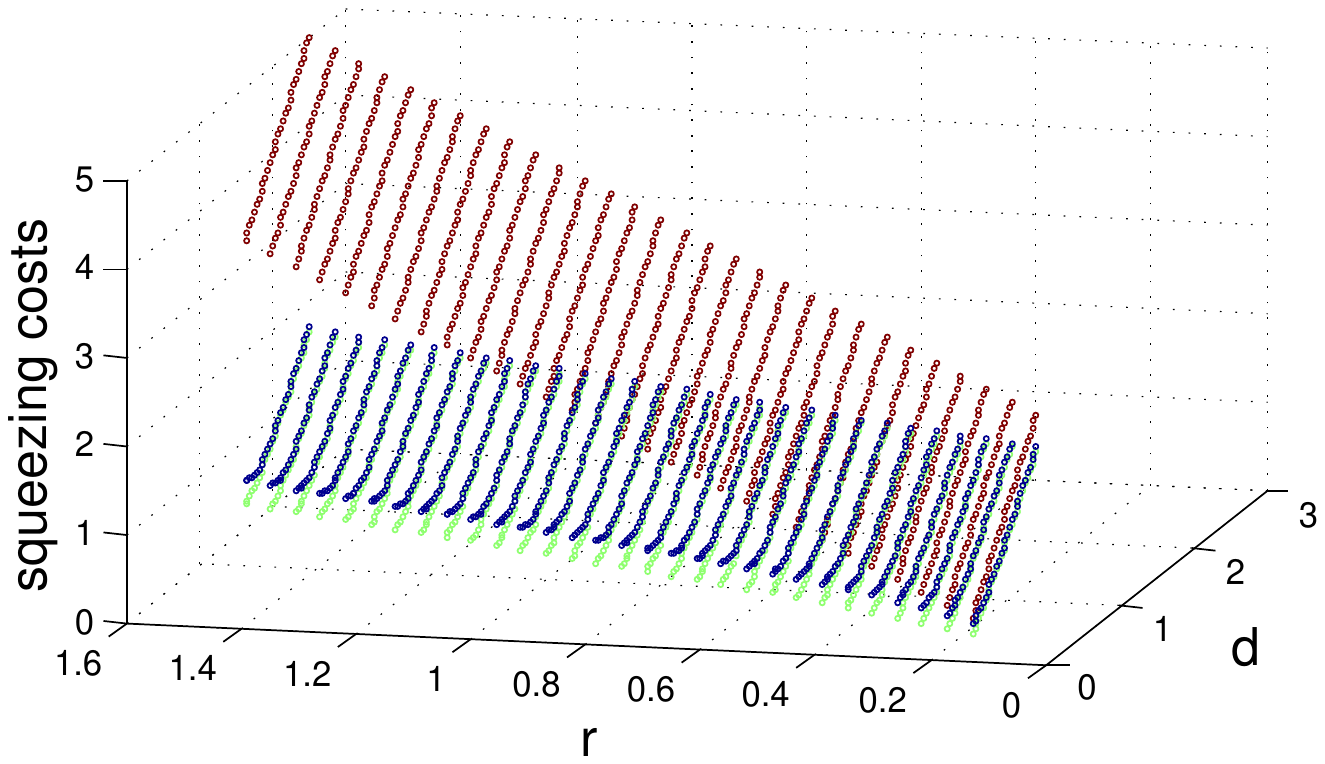} \\
	\includegraphics[trim=2cm 9cm 5cm 13cm,width=12cm]{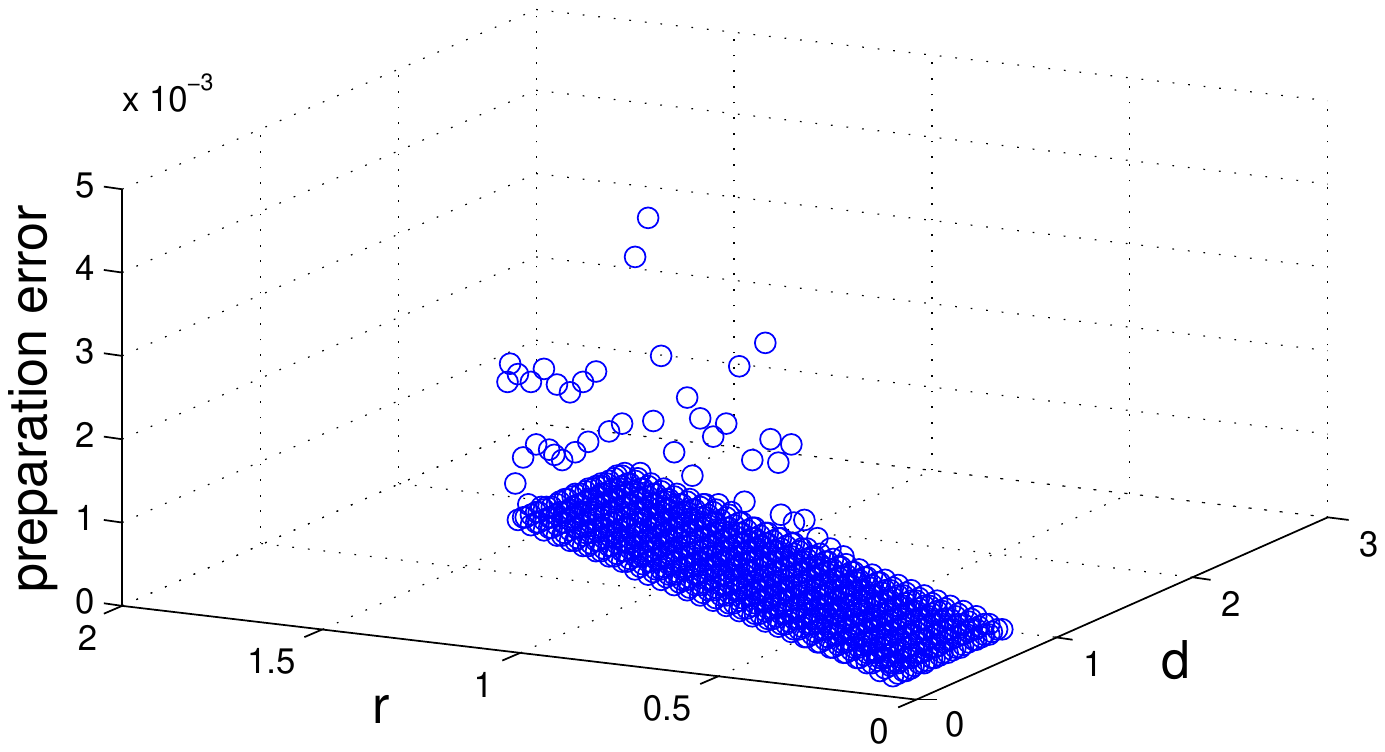}
	\caption{Results of numerical calculations (formulas for $d$ and $r$ in equation (\ref{eqn:examplenumerics})). On the upper figure, the green points are the best lower bound, the blue points denote the value of the objective function at the minimum found by \textsc{SolvOpt} and the red points denote the squeezing costs of the preparation protocol of \cite{mis08} (equation (\ref{eqn:corsqueeze}). The lower figure shows the preparation error. It is mostly below $10^{-6}$.}
\label{fig:korolkova}
\end{figure}

\section{Discussion and open questions} \label{sec:discopen}
We have defined two measures of squeezing - one for quantifying the amount of single-mode squeezing needed to prepare a state and the other for quantifying the amount of squeezed resource states needed to prepare a state. We have seen that both measures give equivalent results and that the measures themselves are continuous on the interior and convex, which gives us a way to calculate them. Let us now discuss applications, open questions and useful modifications while putting the measures into context.

\subsection{Comparison to existing theoretical and experimental squeezing measures}
Squeezed states have been studied experimentally in quantum optical systems for 30 years now and it remains a very challenging endeavour to produce and maintain states with a lot of squeezing. This is the reason why it is interesting to quantify the amount of squeezing necessary for a given task.

In experiments, squeezing of a state is most commonly measured by the logarithm of the smallest eigenvalue. More precisely, given a single-mode pure state $\rho$ with variance $\Delta Q^2\leq \Delta P^2$, the squeezing is measure according to
\begin{align}
	G_{\mathrm{exp}}(\rho)=10\log_{10} 2\langle \Delta Q^2\rangle
\end{align}
and the unit is usually referred to as decibel [dB] \cite{lvo15}. This is (up to a constant which is irrelevant for our purposes since we use different units) equal to the logarithm of the smallest eigenvalue of the covariance matrix $\rho$, if the experimental measure is measured in the basis where the covariance matrix is diagonal. Let us call this measure $G_{\mathrm{minEig}}(\rho):=-\log \lambda_{\mathrm{min}}(\gamma_{\rho})$. Since it is being used in experiments, it has been used and studied also in the theoretical literature, see for instance \cite{kra03a}. We know of no operational interpretation for this measure that is similar to the interpretation given in section \ref{sec:operational}. 

For multi-mode states, there is one very clear drawback to this measure of squeezing: The two states $\rho$ and $\rho^{\prime}$ with covariance matrices 
\begin{align}
	\gamma=\diag(s,s^{-1},1,1), \qquad \gamma^{\prime}=\diag(s,s^{-1},s,s^{-1})
\end{align}
for some parameter $s\geq 1$ have the same smallest eigenvalue $s$ and therefore the same amount of squeezing: $G_{\mathrm{minEig}}(\rho)=G_{\mathrm{minEig}}(\rho^{\prime})=\log(s)$. Especially if we increase the number of modes a lot, this is not very convincing and our measure $G$ has a better behaviour in that respect.

However, there is also a problem regarding our squeezing measure $G$: Squeezing is not just experimentally challenging, it gets much harder if we want to achieve a larger amount of single-mode squeezing. Currently, the highest amount of squeezing obtained in quantum optical systems seems to be about $13\,\mathrm{dB}$ \cite{and15}, which means that if $G_{\mathrm{minEig}}$ exceeds this value, the state cannot be prepared anymore. In other words, the two states $\rho$ and $\rho^{\prime}$ with covariance matrices
\begin{align}
	\gamma=\diag(s,s^{-1},s,s^{-1}), \qquad \gamma^{\prime}=\diag(s^2,s^{-2},1,1) \label{eqn:expprob}
\end{align}
will not be equally hard to prepare although $G(\gamma)=G(\gamma^{\prime})$. This is due to the fact that we quantified the cost of a single-mode squeezer by $\log s$. 

To amend this, one could propose an easy modification to the definition of $F$ in equation (\ref{eqn:defF}):
\begin{align}
	F_g(\gamma)=\sum_{i=1}^n \log(g(s_i^{\downarrow}(S)))
\end{align}
by inserting another function $g:\R\to \R$ to make sure that for the corresponding measure $G_g(\rho)\equiv G_g(\gamma)$, we have $G_g(\gamma)\neq G_g(\gamma^{\prime})$ in equation (\ref{eqn:expprob}). We pose the following natural restrictions on $g$:
\begin{itemize}
	\item We need $g(1)=1$ since $G_g(\rho)$ should be zero for unsqueezed states. 
	\item Squeezing should get harder with larger parameter, hence $g$ should be monoto\-nously increasing.
	\item For simplicity, we assume $g$ to be differentiable. 
\end{itemize}

Let us first consider squeezing operations and the measure $F_g$. We proved in Proposition \ref{prop:foverlinef} and Theorem \ref{thm:liepaths} that $F$ is minimised by the Euler decomposition. A crucial part was given by Lemma \ref{lem:gelfand}, which implies that we cannot simply reduce the costs by splitting the squeezing operation into smaller and smaller parts. In order to be useful for applications, we must require the same to be true for $F_g$, i.e.
\begin{align*}
	\sum_{i=1}^n \log (g(s_i^{\downarrow}(SS^{\prime})))\leq \sum_{i=1}^n [\log (g(s_i^{\downarrow}(S)))+\log(g(s_i^{\downarrow}(S^{\prime})))]
\end{align*}
This puts quite strong restraints on $g$: Considering $n=1$ and assuming that $S$ and $S^{\prime}$ are diagonal with ordered singular values, this implies that $g$ must fulfill $g(xy)\leq g(x)g(y)$ for $x,y\geq 1$. This submultiplicativity restraint rules out all interesting classes of functions: Assume for instance that $g(2)=c$, then $g(2^n)\leq c^n$, where equality is attained if $g(x)=c\cdot x$. Therefore, all submultiplicative functions $g(x)$ for $x\geq 1$ must lie below $g(x)=c\cdot x$ at least periodically - which means that they will grow less than the function $g(x)=c\cdot x$. Hence, Lemma \ref{lem:gelfand} does not hold if we consider increasingly growing functions $g$ that depict the experimental fact that single-mode squeezing is the more challenging the more squeezing we want. This implies that one could at least in some cases make the measure arbitrarily small by splitting the single-mode squeezer into many small single-mode squeezers, which does not help in experiments.

A way to circumvent the failure of Lemma \ref{lem:gelfand} would be to work with the squeezing of formation measure, which uses squeezed state as a resource and does not allow arbitrary splitting. Likewise, one could require that there was only one operation of type (O4) as specified in section \ref{sec:operational} in any preparation procedure. Once again, we would like to have a look at fast-growing functions $g$ such that $G_g$ remains operational. In order to be able to apply the same methods as in the proofs of Theorem \ref{thm:operational}, we need to require the following:
\begin{enumerate}
	\item $\log\circ g\circ \Ca$ is convex on $(1,\infty)$.
	\item $\log(g(\exp(t)))$ is convex and monotone increasing in $t$. 
\end{enumerate}
The first condition replaces the $\log$-convexity of the Cayley transform in the proof of Theorem \ref{thm:convex}, making the measure convex. Using \cite{bha96}, II.3.5 (v), the second condition makes sure that equation (\ref{eqn:needforg}) still holds. The second condition can probably be relaxed while the proof of Theorem \ref{thm:operational} is still applicable. It is intuitive that both conditions hold for convex functions that grow fast enough. The so-defined measures $G_g$ would, given an appropriate $g$, both be operational and they would both reflect that a lot of squeezing in one mode is experimentally hard. A function $g$ fulfilling these prerequisites is $g(x)=\exp(x)$, which would correspond to a squeezing cost increasing linearly in the squeezing parameter. One could even introduce a cutoff after which $g$ would be infinite, reflecting the impossibility of single-mode squeezing beyond a certain amount.

A simpler way to reflect the problems of equation (\ref{eqn:expprob}) would be to consider the measures $G$ and $G_{\mathrm{minEig}}$ together (calculating $G_{\mathrm{minEig}}$ of both the state and the minimal preparation procedure in $G$). Then the former can be used as an overall measure of required squeezing and the latter to determine whether the maximal amount of single-mode squeezing is possible with the equipment used.

A second problem is associated with the form of the Hamiltonian (\ref{eqn:hamilton}). In the lab, the Hamiltonians that can be implemented might not be single-mode squeezers, but other squeezers such as symmetric two-mode squeezers (such as in \cite{scu97}, chapter 2.8). It is clear how to define a measure $T^{\prime}$ for these kinds of squeezers. Clearly, we can always use passive transformations to express any such multi-mode squeezers as several single-mode squeezers, hence $G$ is a lower bound to $G^{\prime}$. We did not further investigate any other set of Hamiltonians, but we believe that the methods developed here for the simplest case of single-mode squeezers can help in developing other cases of experimental interest.

\subsection{The squeezing measure as a resource measure}
In addition to the first definition of $G$ quantifying the amount of single-mode squeezing needed to prepare a state, we also considered squeezed state as a resource and defined the equivalent of the entanglement of formation, the ``squeezing of formation'' in section \ref{sec:resourcemeasure}. We believe that this is the first instance of an operational measure for the resource theory of squeezing in Gaussian states. Together with section \ref{sec:mathmeasure}, we gave an explicit mathematical argument why squeezing can be seen as a resource theory when restricting to the experimentally interesting class of Gaussian states and Gaussian operations. 

A resource theory of squeezing would be interesting from a theoretical perspective, since it is closely linked to the information-theoretically highly relevant resource theory of entanglement (recall that highly entangled states are usually highly squeezed). We provided the first step to investigate this resource theory from an operational perspective and in a systematic way. 

Given the ``squeezing of formation'', one natural further question would be whether ``distillation of squeezing'' is possible with Gaussian operations. It has been shown that in some sense this is impossible for the measure $G_{\mathrm{minEig}}$ in \cite{kra03b}, while it is possible and has been investigated for non-Gaussian states in many papers (cf. \cite{fil13,hee06} and references therein). In our case, it is not immediately clear that it is not possible to extract single-mode squeezed states of less squeezing from a given squeezed state. This and similar questions could be investigated in future work.

\subsection{Open Questions}
Finally, let us list a number of mostly mathematical questions that remain unanswered:
\begin{enumerate}
	\item Is $G$ continuous everywhere? The only cases not covered are jumps along the boundary of the set of covariance matrices.
	\item Is $G$ additive? We know that it is subadditive and we have a good upper bound on superadditivity.
	\item Numerical tests suggest that the lower bounds are pretty good. Can one find simple to calculate good upper bounds?
	\item Is there a simple analytical formula for $G$? Can one give a matrix attaining the minimum?
	\item How does the measure change if we allow different types of basic Hamiltonians as elementary building blocks, for instance symmetric two-mode squeezers instead of single-mode squeezers?
\end{enumerate}

\subsection*{Acknowledgements}
M.I. thanks Konstantin Pieper for discussions about convex optimization and Alexander M\"{u}ller-Hermes for discussions about \textsc{Matlab}. M.I. is supported by the Studienstiftung des deutschen Volkes.
\printbibliography
\appendix
\section{The Cayley trick for matrices} \label{app:cay}
In this appendix, we give an introduction to the Cayley-transform and prove Proposition \ref{prop:cayley}.
\begin{dfn}
Define the Cayley transform and its inverse via:
\begin{align}
\begin{split}
	\Ca:\{H\in\R^{m\times m}|\operatorname{spec}(H)\cap\{+1\}=\emptyset\}\to \R^{m\times m}\\
	H \mapsto \frac{\id+H}{\id-H}
\end{split} \\
\begin{split}
	\Ca^{-1}:\{S\in\R^{m\times m}|\operatorname{spec}(H)\cap\{-1\}=\emptyset\}\to\R^{m\times m} \\
	S\mapsto \frac{S-\id}{S+\id}
\end{split}
\end{align}
\end{dfn}
\begin{lem}
$\Ca$ and $\Ca^{-1}$ are well-defined and inverses of each other. Moreover, $\Ca$ is a diffeomorphism onto its image $\operatorname{dom}(\Ca^{-1})$.
\end{lem}
\begin{proof}
If $\operatorname{spec}(H)\cap\{+1\}=\emptyset$, then $\id-H$ is invertible and $H\mapsto (\id+H)/(\id-H)$ is well-defined, as $[\id+H,\id-H]=0$. Now let $H\in \R^{m\times m}$ be such that $\operatorname{spec}(H)\cap\{+1\}=\emptyset$. We will show that $\Ca(H)$ contains no eigenvalue $-1$. To see this, let
\begin{align}
	H=T\bigoplus_i J(n_i,\lambda_i)T^{-1}
\end{align}
be the Jordan normal form with block sizes $n_i$ and eigenvalues $\lambda_i$. Let us here consider the complex Jordan decomposition, i.e. $\lambda_i$ are allowed to be complex. Then:
\begin{align}
	\id+H=T\bigoplus_i J(n_i,1+\lambda_i)T^{-1}, \quad \id-H=T\bigoplus_i J(n_i,1-\lambda_i)T^{-1}
\end{align}
and thus
\begin{align*}
	\Ca(H)=T\bigoplus_i J(n_i,1+\lambda_i)\cdot J(n_i,1-\lambda_i)^{-1}T^{-1}
\end{align*}
For the inverse of the Jordan blocks, we can use the well-known formula:
\begin{align*}
	\begin{pmatrix}{} 1-\lambda_i & 1 & \ldots & 0 \\
		0 & 1-\lambda_i & \ldots & 0 \\
		\vdots & \vdots & \ddots & \vdots \\
		0 & 0 & \ldots & 1-\lambda_i
	\end{pmatrix}^{-1} 
	=\begin{pmatrix}{} 
		\frac{1}{1-\lambda_i} & \frac{-1}{(1-\lambda_i)^2} & \ldots & \frac{(-1)^{n_i-1}}{(1-\lambda_i)^{n_i}} \\
		0 & \frac{1}{1-\lambda_i} & \ldots & \frac{(-1)^{n_i-2}}{(1-\lambda_i)^{n_i-1}} \\
		\vdots & \vdots & \ddots & \vdots \\
		0 & 0 & \ldots & \frac{1}{1-\lambda_i}
	\end{pmatrix}
\end{align*}
In particular, this is still upper triangular. Then $J(n_i,1+\lambda_i)J(n_i,1-\lambda_i)^{-1}$ is still upper triangular with diagonal entries $(1+\lambda_i)/(1-\lambda_i)$. Since $(1+\lambda_i)/(1-\lambda_i)\neq -1$ for all $\lambda_i\in\C$, we find that $J(n_i,1+\lambda_i)J(n_i,1-\lambda_i)^{-1}$ cannot have eigenvalue $-1$ for any $i$, hence $\operatorname{spec}(\Ca(H))\cap \{-1\}\neq \emptyset$.

Finally, we observe:
\begin{align*}
	\Ca^{-1}\Ca(H)=\frac{\frac{\id+H}{\id-H}-\id}{\frac{\id+H}{\id-H}+\id}
		=\frac{\id+H-\id+H}{\id+H+\id-H}=H
\end{align*}
Moreover, set $f_1(A)=-2A-\id$ for all matrices $A\in\R^{m\times m}$, $f_2(A)=A^{-1}$ for all invertible matrices $A\in\R^{m\times m}$ and $f_3(A)=A-\id$ for all matrices $A\in\R^{m\times m}$. Then we have
\begin{align}
	f_1\circ f_2\circ f_3(H)=f_1\circ	f_2(H-\id)=f_1\left(\frac{1}{H-\id}\right)=-\frac{2}{H-\id}-\id=\Ca(H) \label{eqn:f1f2f3}
\end{align}
Since $f_i$ are differentiable for all $i=1,2,3$, we have that $\Ca$ is invertible.

The same considerations with a few signs reversed also lead us to conclude that $\Ca^{-1}$ is well-defined and indeed the inverse of $\Ca$. We can similarly decompose $\Ca^{-1}$ to show that it is differentiable, making $\Ca$ a diffeomorphism. Here, we define $g_1(A)=2A+\id$ for all $A\in\R^{m\times m}$, $g_2(A)=A^{-1}$ for all invertible $A\in\R^{m\times m}$ and $g_3(A)=A+\id$ for all $A\in\R^{m\times m}$. A quick calculation shows
\begin{align}
	g_1\circ g_2\circ g_3(S)=\Ca^{-1}(S). \label{eqn:g1g2g3}
\end{align}
\end{proof}

Denote by $\mathcal{H}$ the set 
\begin{align}
	\mathcal{H}:=\left\{H=\begin{pmatrix}{} A & B \\ B & -A \end{pmatrix}\middle|A\in\R^{2n\times 2n} A^T=A,B^T=B,-\id<H<\id\right\}
\end{align}
where $H<\id$ means that $\id-H$ is positive definite (not just positive semidefinite). We can then prove the Cayley trick:
\begin{prop}
Let $H\in \R^{2n\times 2n}$. Then $H\in\mathcal{H} \Leftrightarrow (\Ca(H)\in Sp(2n) \wedge \Ca(H)\geq iJ)$.
\end{prop}
\begin{proof}
Note that for $H\in\mathcal{H}$, $1\notin \operatorname{spec}(H)$, hence $\Ca(H)$ is always well-defined. $\Ca(H)=(\id+H)(\id-H)^{-1}\geq 0$, since $\id+H\geq 0$ and $(\id-H)^{-1}\geq 0$ as $-\id<H<\id$. Observe:
\begin{align*}
	HJ=\begin{pmatrix}{} A & B \\ B & -A\end{pmatrix}\begin{pmatrix} 0 & \id \\ -\id & 0 \end{pmatrix}
		=\begin{pmatrix}{} -B & A \\ A & B\end{pmatrix}
		=-\begin{pmatrix} 0 & \id \\ -\id & 0 \end{pmatrix}\begin{pmatrix}{} A & B \\ B & -A\end{pmatrix}
		=-JH.
\end{align*}
Then we can calculate:
\begin{align*}
	(\id+H)\cdot (\id-H)^{-1}J&=-(\id+H)\cdot (J(\id-H))^{-1}=-(\id+H)\cdot ((\id+H)J)^{-1} \\
	&=(\id+H)J(\id+H)^{-1}=J(\id-H)\cdot (\id+H)^{-1},
\end{align*}
hence	$\Ca(H)J=J\Ca(H)^{-1}$ and as $\Ca(H)$ is Hermitian, we have $\Ca(H)^TJ\Ca(H)=J$ and $\Ca(H)$ is symplectic. Via Corollary \ref{cor:symplecticspec}, as $\Ca(H)$ is symplectic and positive definite, we can conclude that $\Ca(H)\geq iJ$. 

Conversely, let $S\in Sp(2n)$ and $S\geq iJ$. Then $S\geq -iJ$ by complex conjugation and $S\geq 0$ after averaging the two inequalities. Since any element of $Sp(2n)$ is invertible, this implies $S>0$. From this we obtain:
\begin{align*}
	\frac{S-\id}{S+\id}&> -\id \qquad \mathrm{as~}S+\id > \id \\
	\frac{S-\id}{S+\id}&< \id \qquad ~~\mathrm{always} 
\end{align*}
Write $(S-\id)\cdot (S+\id)^{-1}=\begin{pmatrix}{} A & B \\ C & D\end{pmatrix}$. As $S$ is Hermitian, $A^T=A$ and $C=B^T$, $D^T=D$. We have on the one hand
\begin{align*}
	\frac{S-\id}{S+\id}J&=(S-\id)\cdot (-S^{-T}J-J)^{-1}=(S-\id)(-J)^{-1}(S^{-T}+\id)^{-1} \\
	&=(SJ-J)\cdot (S^{-T}+\id)^{-1}=J(S^{-T}-\id)S^TS^{-T}(S^{-T}+\id)^{-1} \\
	&=-J\frac{S-\id}{S+\id}
\end{align*}
and on the other hand
\begin{align*}
	\begin{pmatrix}{} A & B \\ B^T & D\end{pmatrix}J&=\begin{pmatrix}{} -B & A \\ -D & B^T \end{pmatrix} \\ -J\begin{pmatrix}{} A & B \\ B^T & D\end{pmatrix}&=\begin{pmatrix}{} -B^T & -D \\ A & B \end{pmatrix}
\end{align*}
Put together this implies $B=B^T$ and $D=-A$, hence $\Ca^{-1}(S)\in \mathcal{H}$, which is what we claimed.
\end{proof}

\begin{prop}
The Cayley transform $\Ca$ is operator monotone and operator convex on the set of $A=A^T\in\R^{m\times m}$ with $\operatorname{spec}(A)\subset (-1,1)$. $\Ca^{-1}$ is operator monotone and operator concave on the set of $A=A^T\in\R^{m\times m}$ with $\operatorname{spec}(A)\subset (-1,\infty)$.
\end{prop}
\begin{proof}
Recall equation (\ref{eqn:f1f2f3}) and the definition of $f_1,f_2,f_3$. $f_1$ and $f_3$ are affine and thus for all $X\geq Y$: $f_3(X)\geq f_3(Y)$ and $f_1(X)\leq f_1(Y)$. For $X\geq Y\geq 0$, we also have $f_2(Y)\geq f_2(X)\geq 0$ since matrix inversion is antimonotone. Now let $-\id\leq Y\leq X \leq 1$, then $-2\id\leq f_3(Y)\leq f_3(X)\leq 0$ and $-1/2\id\geq f_2\circ f_3(Y)\geq f_2\circ f_3(X)\geq 0$ and finally $\Ca(X)\geq \Ca(Y)\geq 0$, proving monotonicity of $\Ca$. Similarly, one can prove that $\Ca^{-1}$ is monotonous using equation (\ref{eqn:g1g2g3}).

For the convexity of $\Ca$, we note that since $f_1,f_3$ are affine they are both convex and concave. It is well-known that $1/x$ is operator convex for positive definite and operator concave for negative definite matrices (to prove this, consider convexity/concavity of the functions $\langle \psi, X^{-1}\psi\rangle$ for all $\psi$). It follows that for $-\id\leq H\leq \id$ we have $f_3(x)\leq 0$, hence $f_2\circ f_3$ is operator concave on $-\id \leq H\leq \id$. As $f_1(A)=-2A-\id$, this implies that $\Ca=f_1\circ f_2\circ f_3$ is operator convex. 

For the concavity of $\Ca^{-1}$, recall equation (\ref{eqn:g1g2g3}) and the definitions of $g_1,g_2,g_3$. Then, given $-\id\leq X$, we have $g_3(X)$ is positive definite and concave as an affine map. $g_2$ is concave on positive definite matrices, as $1/x$ is convex and $(-1)$ is order-reversing, hence $-1/x$ is concave on positive definite matrices. Since $g_1$ is concave as an affine map, $g_1\circ g_2\circ g_3=\Ca^{-1}$ is operator concave for all $-\id\leq X$. 
\end{proof}

\begin{lem}
$\Ca:\R\to \R$ is $\log$-convex on $[0,1)$.
\end{lem}
\begin{proof}
We need to see that the function $h(x)=\log \frac{1+x}{1-x}$ is convex for $x\in[0,1)$. Since $h$ is differentiable on $[0,1)$, this is true iff the second derivative is nonnegative:
\begin{align*}
	h^{\prime\prime}(x)=\frac{4x}{(1-x^2)^2}
\end{align*}
is clearly positive on $[0,1)$ and $h$ is therefore $\log$-convex. 
\end{proof}

\section{Continuity of set-valued functions}\label{app:setvalue}
Here, we provide some definitions and lemmata from set-valued analysis for the reader's convenience. This branch of mathematics deals with functions $f:X\to 2^Y$ where $X$ and $Y$ are topological spaces and $2^Y$ denotes the power set of $Y$. 

In order to state the results interesting to us we define:
\begin{dfn}
Let $X,Y\subseteq \mathbb{R}^{n\times m}$ and $f:X\to 2^Y$ be a set-valued function. Then we say that a function is \emph{upper semicontinuous} (often also called \emph{upper hemicontinuous} to distinguish it from other notions of continuity) at $x_0\in X$ if for all open neighbourhoods $Q$ of $f(x_0)$ there exists an open neighbourhood $W$ of $x_0$ such that $W\subseteq \{x\in X| f(x)\subset Q\}$.

Likewise, we call it \emph{lower semicontinuous} (often called \emph{lower hemicontinuous}) at a point $x_0$ if for any open set $V$ intersecting $f(x_0)$, we can find a neighbourhood $U$ of $x_0$ such that $f(x)\cap V\neq \emptyset$ for all $x\in U$. 
\end{dfn}
Note that the definitions are valid in all topological spaces, but we only need the case of finite dimensional normed vector spaces. Using the metric, we can give the following characterisation of upper semicontinuity:
\begin{lem} \label{lem:uHsc}
Let $X,Y\subseteq \mathbb{R}^{n\times m}$ and $f:X\to 2^Y$ be a set-valued function such that $f(x)$ is compact for all $x$. Then $f$ is upper semicontinuous at $x_0$ if and only if for all $\varepsilon>0$ there exists a $\delta>0$ such that for all $x\in X$ with $\|x-x_0\|< \delta$ we have: for all $y\in f(x)$ there exists a $\tilde{y}\in f(x_0)$ such that $\|y-\tilde{y}\|<\varepsilon$.
\end{lem}
\begin{proof}
$\Rightarrow$: Let $f$ be lower semicontinuous at $x_0$. For any $\varepsilon>0$ the set 
\begin{align}
	B(\varepsilon,f(x_0))\bigcup_{y\in f(x_0)} \{\hat{y}\in Y|\|y-\hat{y}\|<\varepsilon\} \label{eqn:defball}
\end{align}
is an open neighbourhood of $f(x_0)$. Hence there exists an open neighbourhood $W$ of $x_0$, which contains a ball of radius $\delta>0$ such that $B_{\delta}(x_0)\subseteq W\subseteq \{x\in X| f(x)\subset B(\varepsilon,f(x_0))\}$. Clearly this implies the statement.

$\Leftarrow$: Let $Q$ be a neighbourhood of $f(x_0)$. Since $f(x_0)$ is compact this implies that there is a $\varepsilon>0$ such that $B(\varepsilon,f(x_0))\subseteq Q$ where this set is defined as in equation (\ref{eqn:defball}). If this were not the case, for every $n\in \mathbb{N}$ there must be a $y_n\in Y\setminus Q$ such that $\inf_{\hat{y}\in f(x_0)} \|y_n-\hat{y}\|<1/n$. Since by construction this implies that $y_n\in B(1,f(x_0))$, which is compact, a subsequence of these $y_n$ must converge to $y$. As $Y\setminus Q$ is closed as $Q$ is open, $y\in Y\setminus Q$. However, $\inf_{\hat{y}\in f(x_0)}\|y-\hat{y}\|=0$ by construction and since $f(x_0)$ is compact, the infimum is attained, which implies $y\in f(x_0)$. This contradicts the fact that $Q$ is a neighbourhood of $f(x_0)$.

Hence we know that for any open $Q$ containing $f(x_0)$ there exists a $\varepsilon>0$ such that $B(\varepsilon,f(x_0))\subseteq Q$. By assumption, this implies that there exists a $\delta>0$ such that $B_{\delta}(x_0)\subseteq \{x\in X|f(x)\subset B(\varepsilon,f(x_0))\}$. Since clearly $\{x\in X|f(x)\subset B(\varepsilon,f(x_0))\}\subseteq \{x\in X|f(x)\subset Q\}$ we can choose $W:=B_{\delta}(x_0)$ to finish the proof.
\end{proof}
This second characterisation is sometimes called \emph{upper Hausdorff semicontinuity} and it can equally be defined in any metric space. Clearly, the notions can differ for set-valued functions with noncompact values or in spaces which are not finite dimensional. With these two definitions, we can state the following classic result:

\begin{prop}[\cite{dol79}] \label{prop:usc}
Let $Y$ be a complete metric space, $X$ a topological space and $f:X\to 2^Y$ a compact-valued set-valued function. The following statements are equivalent:
\begin{itemize}
	\item $f$ is upper semicontinuous at $x_0$.
	\item for each closed $K\subseteq X$, $K\cap f(x_0)$ is upper semicontinuous at $x_0$.
\end{itemize}
\end{prop}
An interesting question would be whether the converse is also true. Even if $f(x)$ is always convex, this need not be the case if $K\cap f(x_0)$ has empty interior as simple counterexamples can show. In case the interior is nonempty, another classic results guarantees a converse in many cases:
\begin{prop}[\cite{mor75}] \label{prop:lsc}
Let $X$ be a compact interval and $Y$ a normed space. Let $f:X\to 2^Y$ and $g:X\to 2^Y$ be two convex-valued set-valued functions. Suppose that $\operatorname{diam}(f(t)\cap g(t))<\infty$ and $f(t)\cap \operatorname{int}(g(t))\neq \emptyset$ for all $t$. Then if $f,g$ are continuous (in the sense above) so is $f\cap g$.
\end{prop}

\section{Numerical implementation and documentation including source code} \label{app:code}
Here, we provide a short documentation to the program written in \textsc{Matlab}, Version R2014a, and used for the numerical computations in section \ref{sec:examples}. The source Code can be found at GitHub \url{https://github.com/Martin-Idel/operationalsqueezing}.

The program tries to minimise the function $f$ defined in equation (\ref{eqn:convexsum}) over the set $\mathcal{H}$. Throughout, suppose we are given a covariance matrix $\gamma$.

Let us first describe the implementation of $f$: As parameterisation of $\mathcal{H}$, we choose the simplest parameterisation such that for matrices with symplectic eigenvalues larger than one, the set of feasible points has nonempty interior: We parameterise $A,B$ via matrix units $E_i,E_{jk}$ with $i\in\{1,\ldots,n\}$, $k\in\{1,\ldots,n-1\}$ and $j<k$, where $(E_i)_{jk}=\delta_{ij}\delta_{ik}$ and $(E_{jk})_{lm}=\delta_{jl}\delta_{km}+\delta{jm}\delta_{kl}$. This parameterisation might not be very robust, but it is good enough for our purpose. Instead of working with complex parameters, we compute $s_i(A+iB)$ as $\lambda_i^{\downarrow}(H)$ for the matrix 
\begin{align} 
	H=\begin{pmatrix}{} A & B \\ B & -A\end{pmatrix}. \label{eqn:specialform}
\end{align} 
The evaluation of $f$ is done in function \textsc{objective.m}. Since $f$ is not convex for $(A,B)$ with the corresponding $H$ having eigenvalues $\geq 1$ or $\leq -1$, the function first checks, whether this constraint is satisfied and outputs a value that is $10^7$-times larger than the value of the objective function at the starting point otherwise.

The constraints are implemented in function \textsc{maxresidual.m}. Via symmetry, it is enough to check that for any $H$ tested, $\lambda_{2n}^{\downarrow}(H)\geq 1$. The second constraint is given by $\mathcal{C}^{-1}(\gamma)\geq H$ and this is tested by computing the smallest eigenvalue of the difference. 

The function which is most important for users is \textsc{minimum.m}, which takes a covariance matrix $\Gamma\geq iJ$, its dimensions $n$ and a number of options as arguments and outputs the minimum. Note that the program checks whether the covariance matrix is valid. For the minimisation, we use the \textsc{Matlab}-based solver \textsc{SolvOpt} (\cite{kun97}, latest version 1.1). \textsc{SolvOpt} uses a subgradient based method and the method of exact penalization to compute (local) minima. For convex programming, any minimum found by the solver is therefore an absolute minimum. In order to work, the objective function may not be differentiable on a set of measure zero and it is allowed to be nondifferentiable at the minimum. Since $f$ is differentiable for all $H$ with nondegenerate eigenvalues, this condition is met. In addition, \textsc{SolvOpt} needs $f$ to be defined everywhere, as it is not an interior point method. Since $f$ is well-defined but not convex for $H\notin\mathcal{H}$ and $\operatorname{spec}(H)\cup \{1\}=\emptyset$, we remedy this by changing the output of \textsc{objective.m} to be very large when $H\notin \mathcal{H}$ as described above. Constraints are handled via the method of exact penalisation. We used \textsc{SolvOpt}'s algorithm to compute the penalisation functions on its own. 

It is possible (and for speed purposes advisable) to implement analytical gradients of both the objective and the constraint functions. Following \cite{mag85}, for diagonalisable matrices $A$ with no eigenvalue multiplicities, the derivative of an eigenvalue $\lambda_i(A)$ is given by:
\begin{align}
	\partial_E\lambda_i(A)=v_i(A)^T\partial_EAv_i(A) \label{eqn:eigderiv}
\end{align}
where $v_i(A)$ is the eigenvector corresponding to $\lambda_i(A)$ and $\partial_v(A)=\lim\limits_{h\to 0} (A+hE-A)/h=E$. Luckily, if $A$ is not differentiable, this provides at least one subgradient. An easy calculation shows that a subgradient of the objective function $f$ for matrices $H$ with $-\id<H<\id$ in the parameterisation of the matrix units $E_{ij}$ is given by
\begin{align}
	(\nabla f)_i=\sum_{j=1}^n\frac{\partial_i \lambda_j^{\downarrow}(H)}{(1+\lambda_j(H))(1-\lambda_j(H))^2}
		=\sum_{j,k=1}^n\frac{v^T_{j,k}F(i)v_{k,j}}{(1+\lambda_j(H))(1-\lambda_j(H))^2}
\end{align}
with $F$ being the matrices corresponding to the chosen parameterisation. The gradient of the constraint function is very similar and given by equation (\ref{eqn:eigderiv}) for $A=\gamma-H$ or $A=2\id-H$ depending on which constraint is violated. This is implemented in functions \textsc{objectivegrad.m} and \textsc{maxresidualgrad.m}.

\textsc{SolvOpt} needs a starting point. Given $\Gamma$, via Williamson's Theorem, $\Gamma=S^TDS\geq S^TS$, hence $S^TS$ provides a good starting point. The function \textsc{williamson.m} computes the Williamson normal form for $\gamma$ and returns $S$, $D$ and $S^TS$, the latter of which is used as starting point. It computes $S$ and $D$ essentially by computing the Schur decomposition of $\Gamma^{-1/2}J\Gamma^{-1/2}$ (in the $\sigma$-basis instead of the $J$-basis). $S$ is then given by $S^T=\gamma_{1/2}KD^{-1/2}$ (see the proof of \cite{sim99}), where $K$ is the Schur transformation matrix. 

A number of comments are in order:
\begin{enumerate}
	\item All functions use global variables instead of function handles. This is required by the fact that \textsc{SolvOpt} has not been adapted to the use of function handles. The user should therefore always reset all variables before running the program.
	\item \textsc{SolvOpt} is not an interior point method, i.e. the results can at times violate constraints. We use the default value for the accuracy of constraints, which is $10^{-8}$ and can be modified by option six. The preparation error should be of the same order than the accuracy of constraints as long as the largest eigenvalue of the minimising symplectic matrix is of order one. 
	\item For our numerical tests, we used bounds on the minimal step-size and the minimal error in $f$ (\textsc{SolvOpt} options two and three) of the order $10^{-6}$ and $10^{-8}$, which seemed sufficient.
	\item All functions called by \textsc{SolvOpt} (the functions \textsc{objective.m}, \textsc{objectivegrad.m}, \textsc{maxresidual,m}, \textsc{maxresidualgrad.m} and \textsc{xtoH.m}) are properly vectorised to ensure maximal speed. 
\end{enumerate}

Finally, \textsc{bounds.m} contains all lower- and upper bounds described in section \ref{sec:bounds}. The semidefinite programme was solved using \textsc{CVX} (version SDPT3 4.0), a toolbox developed in \textsc{Matlab} for disciplined convex programming including semidefinite programming \cite{gra08}. The third bound is not described in section \ref{sec:bounds} - it is an iteration of Corollary \ref{cor:superadd} assuming superadditivity, hence in principle it could be violated. If it were violated, this would immediately disprove superadditivity, which has never been observed in our tests.

\paragraph*{Issues and further suggestions:}
It occurs sometimes that the algorithm does not converge to a minimum inside or near the feasible set. We believe that this is due to instabilities in the parameterisation and implementation. The behaviour can occur while using numerical as well as analytical subgradients, although it occurs more often with analytical ones. For every example where we could observe a failure with either numerical or analytical subgradients, one other method (using numerical subgradients, using analytical subgradients or a mixture thereof) worked fine. In cases of failure, the routine issued several warnings and the result usually lies below the lower bound. A different type of implementation might lead to an algorithm that is more stable, but we did not pursue this any further. It might also be worth to consider trying to compute the penalty function analytically.

In terms of performance times, the algorithm is generally fast for small numbers of modes. When analytical subgradients are not implemented, the performance bottleneck is given by the functions \textsc{xtoH.m}, which is called most often. When analytical subgradients are provided, the performance is naturally much faster. This is particularly important when the number of modes increases. While for five modes, the calculation is done within seconds, already for ten modes and depending on the matrix, it can take a minute on a usual laptop (the algorithm now takes the most amount of time for eigenvalue computations, which seems unavoidable). For even larger matrices, it might be advisable to switch from using the Matlab function \textsc{Eig} to \textsc{Eigf}, but for our examples this did not lead to a time gain.

\end{document}